\documentclass[onefignum,onetabnum]{siamonline250211}

\usepackage{lipsum}
\usepackage{amsfonts}
\usepackage{graphicx}

\usepackage[skip=0pt]{caption}

\usepackage{epstopdf}
\usepackage{algorithmic}
\ifpdf
  \DeclareGraphicsExtensions{.eps,.pdf,.png,.jpg}
\else
  \DeclareGraphicsExtensions{.eps}
\fi

\usepackage{enumitem}
\setlist[enumerate]{leftmargin=.5in}
\setlist[itemize]{leftmargin=.5in}

\usepackage{multirow}

\newcommand\bR{\mathbb{R}}

\newcommand\gt{\hat{x}}
\newcommand{\norm}[2]{\|#1\|_{#2}}
\DeclareMathOperator{\spn}{span}
\newcommand{\orth}{P_{\Sigma}^{\perp} }
\newcommand{\orthi}{P_{V_i}^{\perp} }

\usepackage{makecell}
\usepackage{graphicx}
\usepackage{subcaption}

\usepackage{makecell}
\usepackage{booktabs}
\usepackage{placeins}

\newcommand\myeq{\mathrel{\stackrel{\makebox[0pt]{\mbox{\normalfont\tiny def}}}{=}}}

\newsiamremark{remark}{Remark}
\newsiamremark{hypothesis}{Hypothesis}
\crefname{hypothesis}{Hypothesis}{Hypotheses}
\newsiamthm{claim}{Claim}
\newsiamremark{fact}{Fact}
\crefname{fact}{Fact}{Facts}

\headers{Stochastic Orthogonal Regularization for deep projective priors}{A. Joundi, Y. Traonmilin, A. Newson}

\title{Stochastic Orthogonal Regularization for deep projective priors\thanks{Submitted to the editors DATE.
}}

\author{Ali Joundi\thanks{Université de Bordeaux, CNRS, Bordeaux INP, IMB, UMR 5251
		33400 Talence, France 
  (\email{ali.joundi@u-bordeaux.fr}).}
\and Yann Traonmilin\thanks{Université de Bordeaux, CNRS, Bordeaux INP, IMB, UMR 5251
	33400 Talence, France 
	(\email{yann.traonmilin@u-bordeaux.fr}).} 
\and Alasdair Newson\thanks{ISIR, Sorbonne Université, Paris, France} (\email{anewson@isir.upmc.fr}).}
\usepackage{amsopn}

\ifpdf
\hypersetup{
  pdftitle={Stochastic Orthogonal Regularization for deep projective priors},
  pdfauthor={A. Joundi, Y. Traonmilin, A. Newson}
}
\fi

\begin{document}

\maketitle

\begin{abstract}
Many crucial tasks in image processing and computer vision are formulated as inverse problems. Therefore, it is of great importance to design efficient and robust algorithms to solve these problems. In this paper, we focus on generalized projected gradient descent (GPGD) algorithms where generalized projections are realized with learned neural networks as they provide state-of-the-art results for imaging inverse problems. Indeed, neural networks allow for projections onto unknown low-dimensional sets that model complex data, such as images. We call these projections deep projective priors. In generic settings, when the orthogonal projection onto a low-dimensional model set is used, it has been shown, under a restricted isometry assumption, that the corresponding orthogonal PGD converges with a linear rate, yielding near-optimal convergence (within the class of GPGD methods) in the classical case of sparse recovery. However, for deep projective priors trained with classical mean squared error losses, there is little guarantee that the hypotheses for linear convergence are satisfied. In this paper, we propose a stochastic orthogonal regularization of the training loss for deep projective priors. This regularization is motivated by our theoretical results: a sufficiently good approximation of the orthogonal projection guarantees linear stable recovery with performance close to orthogonal PGD. We show experimentally, using two different deep projective priors (based on autoencoders and on denoising networks), that our stochastic orthogonal regularization yields projections that improve convergence speed and robustness of GPGD in challenging inverse problem settings, in accordance with our theoretical findings.

\end{abstract}

\section{Introduction}

An imaging inverse problem aims to recover an original image $\hat{x}\in \mathbb{R}^n$ from an observation $y\in \mathbb{R}^m$. We model this problem with the observation model:
\begin{equation}\label{eq:eq_pb_inverse}
	y=\mathbf{A}\hat{x}+e,
\end{equation}

where $\mathbf{A}\in \mathbb{R}^{m\times n}$ is an underdetermined measurement matrix ($m<n$) and $e$ is a measurement noise (e.g. $e \sim \mathcal{N}(0,\sigma^2\mathbf{I})$). For a stable estimation of $\hat{x}$ to be possible, a low-dimensional model on $\hat{x}$ is necessary. 
We formalize this by supposing that $\hat{x}\in\Sigma$, where $\Sigma$ is a low-dimensional model set (sets of sparse vectors are typical examples of such low-dimensional model sets).

In this paper, we focus on a specific class of algorithms for estimating $\hat{x}$: Generalized Projected Gradient Descent (GPGD) methods. These algorithms are described by the following iterations, given an initialization $x_0\in\mathbb{R}^n$:
\begin{equation}\label{eq:pgd}
	x_{k+1}=P_{\Sigma}(x_k)-\gamma \mathbf{A}^T(\mathbf{A}P_{\Sigma}(x_k)-\mathbf{y})),
\end{equation}
where $\gamma \geq 0$ is the step size and $P_{\Sigma}$ is a generalized projection operator that maps inputs from $\bR^n$ into the model set $\Sigma$. Recent work has proposed to carry out these projections onto $\Sigma$ via learned neural networks, yielding state-of-the-art results in solving inverse problems~\cite{dong2014learning,kamilov2017plug,zhang2018ffdnet,gupta2018cnn}. We refer to networks used in this manner as \emph{Deep Projective Priors}. More formally, we define a deep projective prior (DPP) as a function $P:\mathbb{R}^n \rightarrow \Sigma$, parametrized by a deep neural network and learned on a database $X$. In this case, the model $\Sigma$ is the set of fixed points of $P$. In this article, we will consider two variations of such priors: 
\begin{itemize}
	\item DPP's parametrized by \emph{autoencoders} that are the composition $P =f_D\circ f_E$ of an encoder $f_E:\mathbb{R}^n \rightarrow \mathbb{R}^d$ and a decoder $f_D:\mathbb{R}^d \to \mathbb{R}^n $. In this case, the dimensionality of the low-dimensional model $\Sigma$ is set explicitly through the latent space size $d$;
	\item DPP's parametrized by \emph{denoisers} (plug-and-play DPP): $P =D$ where  $D$ is a denoiser, i.e. for elements of $\hat{x}\in\Sigma$ some noise $\varepsilon$ (not related to $e$), we have $D(\hat{x}+\varepsilon) \approx \hat{x}$.
\end{itemize}
To train a DPP on a dataset $X$, we minimize a loss function $\mathcal{L}$ that depends on the type of prior:
\begin{equation}\label{eq:eq_ae}
	\begin{split}
		\mathcal{L}_{X,\text{AE}}(P) =\mathcal{L}_{X,\text{AE}}(f_D\circ f_E):=\sum_{x\in X} \norm{f_D\circ f_E(x)-x}{2}^2,~\text{for an autoencoder},\\
		~ \mathcal{L}_{X,\text{PnP}}(P) =\mathcal{L}_{X,\text{PnP}}(D):=\mathbb{E}_{x\in X, \varepsilon\sim\mathcal{N}(0,\xi^2\mathbf{I})}\left(\|D(x+\varepsilon)-x\|_2^2\right),~\text{for a denoiser}.
	\end{split}
\end{equation} 
where $f_E$ is an encoder, $f_D$ is a decoder, $D$ a denoiser and $\varepsilon \in \mathcal{N}(0,\xi^2\mathbf{I})$ is a noise.

It has recently been shown~\cite{traonmilin2024towards} that the GPGD converges \emph{linearly} under a restricted isometry property of $\mathbf{A}$ and the existence of a \emph{restricted} Lipschitz constant of the projection. In particular, the \emph{orthogonal projection} has been shown to have a restricted Lipschitz constant which guarantees this linear convergence, and, furthermore, yields near-optimal convergence when $\Sigma$ is a set of sparse vectors. Unfortunately, for DPPs trained with losses~\eqref{eq:eq_ae}, there is little guarantee that the hypotheses (the restricted Lipschitz property) for linear convergence are satisfied. In this paper, we ask ourselves how to improve the training of DPPs with respect to theoretical convergence guarantees.

\subsection{Contributions} 

We propose a Stochastic Orthogonal Regularization (SOR)\footnote{Not to be confused with Successive Over-Relaxation for solving linear systems of equations} loss that encourages the projection, represented by a DPP, to be orthogonal. In Section~\ref{section:theory}, we introduce the relevant theoretical framework and prove that when $P_\Sigma$ is approximately orthogonal, we can provide performance guarantees close to the orthogonal PGD: i.e., linear convergence under a restricted isometry property of the operator $\mathbf{A}$. These guarantees are the basis for the introduction of stochastic orthogonal regularization in Section~\ref{section:section_ortho}. In Section~\ref{section:exp}, we implement and test SOR for two different DPPs: autoencoders and denoisers and we compare the recovered solutions of noisy inverse problems with and without SOR. In Section~\ref{section:ae}, we first illustrate SOR's results on the MNIST dataset, when using a standard autoencoder over two inverse problems: the inpainting of missing pixels and the image super-resolution. We then move on, in Section~\ref{section:pnp-celebA}, to a more real-world setting: the CelebA dataset, where we use a DRUNET denoiser as the DPP. Our experiments show that SOR yields a faster recovery with minimal deterioration of the PSNR. It also gives solutions which are robust to very challenging conditions with respect to the noise and measurement matrix (eg, with many missing pixels).

\subsection{Related Work:} 

Projected gradient descent algorithms for inverse problems have been studied in multiple contexts. For compressed sensing and sparsity models, iterative hard thresholding has been proposed ~\cite{jain2014iterative}. It uses the exact orthogonal projection and converges linearly to a stable estimation of $\hat{x}$ under a restricted isometry property. For more general low-dimensional models, the linear convergence of this algorithm has been shown in \cite{blumensath2011sampling} under this same property. \cite{golbabaee2018inexact} considers approximated projection with associated convergence theory. We focus here on PGD with a generalized notion of projection (which is not necessarily orthogonal). \cite{traonmilin2024towards} shows that this linear convergence is still valid under a RIP and a restricted Lipschitz property over the projection operator. It is observed experimentally that plug-and-play methods using generalized PGD iterations exhibit this linear convergence. 

Using deep neural networks to solve an inverse problem has been widely explored. Some techniques rely on training a precise deep neural network to solve an inverse problem without specifically knowing the model (\ref{eq:eq_pb_inverse}), \cite{dong2014learning, zhang2018ffdnet} or via unrolling of an optimization method \cite{diamond2017unrolled}. The problem with these techniques is that the associated deep neural networks must be retrained if the measurement operator changes. Other methods use them within an existing optimization algorithm. \cite{scarlett2022theoretical} and \cite{ongie2020deep} give an overview of several techniques with some associated theoretical results. It is in particular possible to use, in a PGD algorithm, explicit projections via autoencoders \cite{solving_inv_p}, or generative models such as Variational Autoencoders \cite{gonzalez2022solving} or GANs \cite{shah2018solving}. On the other hand, Plug and play (PnP) has become widely used in recent years as it uses deep neural network denoisers to perform the iteration of an optimization algorithm. The pretrained denoisers can be used in several optimization algorithms. This flexibility justifies the “Plug and play” terminology. Methods such as PnP-ADMM \cite{venkatakrishnan2013plug,chan2016plug}, PnP-Primal Dual \cite{meinhardt2017learning} , PnP-Half Quadratic Splitting \cite{zhang2021plug}, and PnP-(Fast) Shrinkage Thresholding algorithm (FISTA) \cite{sun2019online} have been proposed. Other works link denoisers with a specific variational prior \cite{romano2017little,cohen2021regularization,renaud2024plug}. 

Diffusion-based methods for image restoration have been widely explored recently. They use the reverse diffusion process to solve a given inverse problem. For instance, \cite{kadkhodaie2021stochastic} considers the prior provided by a trained denoiser. Other methods \cite{chung2022come,chung2022improving,chung2022diffusion,kawar2022denoising,delbracio2023inversion} rely on a diffusion model to recover the original image by taking into account the data consistency during this process. Other generative models such as flow-matching have been considered as priors~\cite{zhang2024flow}.  \cite{martin2024pnp} implemented flow priors as a plug-and-play approach  that alternates between gradient step, projection on flow trajectories, and denoising to recover the original image. Broader overviews are proposed in \cite{he2025diffusion, daras2024survey}.  Note that these methods rely on stochastic algorithms that are out of the scope of our theoretical setting. In this paper, we specifically focus on PnP methods that can be described by deterministic GPGD iterations such as proximal gradient methods (PGM) in \cite{kamilov2017plug}. Linear convergence of GPGD has been proven under restrictive hypotheses in \cite{liu2021recovery} that have been weakened in \cite{traonmilin2024towards}.

Hypotheses guaranteeing convergence of GPGD are not enforced during training with classical MSE loss. Hence, to obtain such guarantees, the loss can be penalized to have these desired properties. For example, several articles propose to control the global Lipschitz condition of neural networks. \cite{gouk2021regularisation} uses a projected stochastic gradient method to optimize a neural network to have a Lipschitz constant for each layer, whereas \cite{miyato2018spectral} proposes to normalize the spectral norm of each layer's weights, ensuring that the Lipschitz constant of the network is upper bounded by one. \cite{ryu2019plug} uses this last proposed regularization for PnP ADMM and PnP forward-backward splitting and proves experimentally the convergence of these algorithms under the Lipschitz condition of the denoisers (achieved by the regularization).\cite{pesquet2021learning} penalizes the norm  of the gradient of a learned operator to approximate the resolvent of a maximally monotone operator.  \cite{hurault2022proximal} trains a denoiser as a gradient descent step, ensuring the convergence of these PnP algorithms in this context.  The Stochastic Orthogonal Regularization loss proposed in this paper encourages a DPP to approximate an orthogonal projection to constrain the function to have a (weaker than global) restricted Lipschitz constant that is still sufficient to improve the rate of convergence. Note that it must not be confused with the orthogonalization of weights of the convolution matrices as in \cite{liu2017deep,wang2020orthogonal} nor the orthogonalization of feature space of the neural network as in \cite{ranasinghe2021orthogonal,wang2019clustering}: it does not need to act directly on the neural network parameters.

\section{Stochastic orthogonal regularization and linear convergence of generalized projected gradient descent}\label{section:theory}
We recall in this section key results for the convergence of generalized projected gradient descent. We show that an approximate orthogonal projection has a restricted Lipschitz constant approximating the constant of the orthogonal projection, i.e. the main hypothesis for convergence is verified for approximate orthogonal projections. These results motivate us to propose our stochastic orthogonal regularization scheme. All proofs are given in the appendix.
\subsection{Definitions}

The quality of a measurement matrix is assessed with a \emph{restricted isometry constant}: 

\begin{definition}\label{def:RIC}
	An operator $\mathbf{B}$ has Restricted Isometry Constant (RIC) $\delta <1$ on the secant set $\Sigma-\Sigma =\{x_1-x_2 : x_1,x \in \Sigma \}$ if for all $x_1,x_2 \in \Sigma$, we have
	
	\begin{equation}
		\|(\mathbf{I}-\mathbf{B})(x_1-x_2)\|_2\leq \delta \|x_1-x_2\|_2
	\end{equation}
	We write $\delta_\Sigma(\mathbf{B})$ the smallest admissible restricted isometry constant (RIC).
\end{definition}

In this article, we will consider RIC $\delta_\Sigma(\gamma \mathbf{A}^T \mathbf{A})$ i.e. the restricted conditioning of $\mathbf{I}-\gamma \mathbf{A}^T \mathbf{A}$, where $\gamma$ is the GPGD step size.
Having a restricted isometry constant (RIC) implies that $\mathbf{A}$ has a restricted isometry property, a typical hypothesis in compressive sensing and sparse recovery (see \cite{Foucart_2013} for an overview).

We introduce our notions of generalized projection and an orthogonal projection. 

\begin{definition}[Generalized projection]
	Let $\Sigma \subset \mathbb{R}^n$. A generalized projection onto $\Sigma$ is a (set-valued) function $P$ such that for any $z\in\bR^N$, $P(z)\subset \Sigma$.
\end{definition}

To simplify notations, $P(z)$ refers to any $u\in P(z)$. While projections realized by deep neural networks usually return a single value, in sparse recovery this is not the case.

\begin{definition}[Orthogonal projection]\label{def:orth_proj}
	We define the orthogonal projection onto a set $\Sigma \subset \mathbb{R}^n$ as follows: for all $z\in\mathbb{R}^n$
	\begin{equation}\label{eq:def_orth_proj}
		P_\Sigma^\perp(z) = \arg \min_{x\in\Sigma} \| x-z \|_2.
	\end{equation}
\end{definition}

We will suppose in the following that $\Sigma$ is a proximinal set. A proximinal set is a set for which the orthogonal projection can always be defined through minimization~\eqref{eq:def_orth_proj}, i.e. $P_\Sigma^\perp(z) \neq \emptyset$ for all $z \in\mathbb{R}^n$. We introduce the restricted Lipschitz condition of a generalized projection onto the model set $\Sigma$. 

\begin{definition}[Restricted Lipschitz property]\label{def:lip_const}
	Let $P$ be a generalized projection. Then $P$ has the restricted $\beta$-Lipschitz property with respect to $\Sigma$ if for all $z \in \mathbb{R}^n, x \in \Sigma, u \in P(z) $ we have
	\begin{equation}
		\begin{split}
			\|u-x\|_2 &\leq \beta \|z-x\|_2\\
		\end{split}
	\end{equation}
	We denote by $\beta_{\Sigma}(P)$ the smallest $\beta$ such that $P$ has a restricted $\beta$-Lipschitz property.
\end{definition}

This is, as its name implies, a constrained version of the general Lipschitz condition where $x$ is restricted to $\Sigma$ (as $P(x)=x$ if $\beta < +\infty$). Note that, while orthogonal projections on proximinal sets are restricted $2$-Lipschitz \cite{traonmilin2024towards}, they are not globally Lipschitz in general (e.g. when $\Sigma = \Sigma_k$ the set of $k$-sparse vectors and $P_\Sigma^\perp$ is the hard thresholding operator). 

The RIC and the restricted $\beta$-Lipschitz condition are sufficient conditions for the linear convergence of GPGD. We extend the convergence result of \cite[Theorem 2.1]{traonmilin2024towards} that was given in the noiseless case to the noisy observation model \eqref{eq:eq_pb_inverse}. 

\begin{theorem}[Stable linear recovery of low-dimensional models]\label{th:gen_convevergence}
	Consider the observation model (\ref{eq:eq_pb_inverse}).
	Let $\Sigma \subset \mathbb{R}^n$. Suppose $\gamma \mathbf{A}^T\mathbf{A}$ has a RIC $\delta:=\delta_\Sigma(\gamma \mathbf{A}^T\mathbf{A})$. Suppose $P_\Sigma$ is a generalized projection with the restricted $\beta$-Lipschitz property with respect to $\Sigma$ (with $\beta := \beta_\Sigma(P_\Sigma)$). Suppose $\delta\beta <1$.
	For any $\hat{x} \in \Sigma, x_0 \in\mathbb{R}^n$, consider the iterates $x_n$ resulting from GPGD iterations~\eqref{eq:pgd}, we have 
	\begin{equation}
		\begin{split}
			\|x_{n} -\gt\|_2& \leq  (\delta \beta)^{n}\|x_0-\gt\|_2+ \gamma \left(\sum_{i=0}^{n-1}(\delta \beta)^{i}\right) \| \mathbf{A}^T e\|_2 \\
			& \leq (\delta \beta)^{n}\|x_0-\gt\|_2 + \frac{\gamma}{1-\delta \beta}\| \mathbf{A}^T e\|_2.
		\end{split}
	\end{equation}	
\end{theorem}

For a fixed measurement matrix $\mathbf{A}$ and step size $\gamma$, the convergence rate and the stability to noise depend solely on the restricted $\beta$-Lipschitz constant. For increased convergence rate and stability to noise, the rate parameter $r=\delta\beta$ must be lowered. Hence, a lower $\beta$ will improve convergence guarantees. This shows in particular the importance of choosing a suitable projection onto $\Sigma$. The orthogonal projection is a good candidate as its Lipschitz constant is bounded $\beta_\Sigma(P_\Sigma^\perp) \leq 2$. In the case of sparse recovery it was shown that $\beta_\Sigma(P_\Sigma^\perp) \leq \sqrt{\frac{3 + \sqrt{5}}{2}} \approx 1.618$ and that this restricted Lipschitz constant is near-optimal \cite{traonmilin2024towards}. 

Our objective is to build deep projective priors that approximate orthogonal projections in order to have good control of the restricted Lipschitz constant. We formalize this idea in the next section.

\subsection{Restricted Lipschitz constants of approximate orthogonal projections }

 We now introduce quantities that control the difference between a given generalized projection and the orthogonal projection.

 \begin{theorem}\label{theorem:general_rlc}
 	Let $\Sigma$ be a proximinal set. Let $\beta^\perp := \beta_\Sigma (P_\Sigma^\perp)$. Let $P$ be a generalized projection onto $\Sigma$. Let $L\geq 0$ such that:
 	$\norm{P(z)-\orth(z)}{2} \leq L \norm{z-\orth(z)}{2}$ for all $z\in\bR^n$. Then $P$ has a restricted Lipschitz constant $\beta=\beta^{\perp}+L$. 
 \end{theorem}
 
 This Theorem implies that $\beta_\Sigma(P) = \beta^{\perp} $ for $P=\orth$ and $\beta_\Sigma(P)$ increases as $P$ deviates from $\orth$ \textit{i.e.} if there exists a $z\in\mathbb{R}^n$ such that $\norm{P(z)-\orth(z)}{2}^2>>\norm{z-\orth(z)}{2}^2$. Hence, even if we do not have explicit access to an orthogonal projection, we can still ensure a linear convergence rate that can be improved by decreasing $L$ with the convergence rate parameter 
 \begin{equation}
 	r= \delta\left(\beta^{\perp}+L\right).
 \end{equation} 
 
 This is achieved by minimizing the constant $L$. Unfortunately, $L$ depends on $\orth$, which is generally unknown in a deep learning setting. In Theorem \ref{theorem:eta_L}, we propose to address this issue by decomposing $L$ into two constants, such that one of these constants is independent of $\orth$. We suppose in particular that $\Sigma$ is a homogeneous model \textit{i.e.} for all $x\in\Sigma$ and $c\in\bR$, $cx\in\Sigma$, that commonly models images in the context of compressive sensing and sparse recovery (\textit{eg.} modifying the brightness or the contrast of an image does not change its nature and the object it represents.)

 For a given generalized projection $P$ onto $\Sigma$, let us define the function $\psi_P:\mathbb{R}^n\rightarrow \mathbb{R}^+$:
 \begin{equation}\label{eq:psi}
 	\psi_P(z):=\frac{|\langle P(z),z-P(z)\rangle|}{\norm{P(z)}{2}\norm{z-P(z)}{2}}.
 \end{equation}
 This function quantifies the orthogonality between the projection $P$ of a given point $z \in \mathbb{R}^n$ and the projection direction: consider the function $\alpha(x,y):=\langle \frac{x}{\norm{x}{2}}, \frac{y}{\norm{y}{2}} \rangle$ which is the cosine of the angle between vectors $x$ and $y$, then $	\psi_P(z) = \alpha(P(z),z-P(z)) $.
 
 \begin{theorem}\label{theorem:eta_L}
 	Let $\Sigma$ be a homogeneous proximinal set and $P$ a generalized projection onto $\Sigma$ such that $P(z) = z$ for $z \in \Sigma$. Let $L':=\underset{z\in\mathbb{R}^n\setminus \Sigma}{\sup}\left(\frac{\norm{\orth(z)-P(z)}{2}}{\norm{\orth(z)-z}{2}}\right)$, $\Psi_P:=\underset{z\in\mathbb{R}^n\setminus \Sigma}{\sup} \psi_P(z)$ and 
 	$\Phi_P:= \underset{z\in\mathbb{R}^n\setminus \Sigma}{\sup}\left(\frac{2\sqrt{1-\alpha(P(z),\orth(z))^2}}{1-\alpha(z,\orth(z))^2}\right)^{\frac{1}{2}}$. Suppose $\Psi_P^2+\Phi_P^2<1$. Then the generalized projection $P$ is $\beta^\perp + L'$ restricted Lipschitz with
 	\begin{equation}
 		L'\leq \frac{\Psi_P}{\sqrt{1-\Psi_P^2-\Phi_P^2}} + \Phi_P.
 	\end{equation}

 \end{theorem}

 This theorem states that we can control the constant $L$ (and consequently the restricted $\beta$-Lipschitz constant) by controlling the following quantities:
 \begin{itemize}
 	\item The constant $\Phi_P$ can be interpreted as the supremum over $z$ of an angular distance between the subspace spanned by $P(z)$ and the subspace spanned by $P_\Sigma^\perp(z)$. Indeed, the quantity $1-\alpha(P(z),\orth(z))^2 \approx 0$ when this angle is close to $0$.
 	\item The constant $\Psi_P$ can be interpreted as an intrinsic orthogonality parameter of the generalized projection $P$, i.e. $\Psi_P = 0$ implies that $\langle P(z), P(z)-z \rangle=0$ for all $z$. The knowledge of $\Sigma$ is not necessary to calculate the local orthogonality $\psi_P(z)$.
 \end{itemize}

 Most importantly, note that these quantities depend on the behavior of the generalized projection $P$ outside of the model set $\Sigma$. Hence, we cannot hope for a control of these quantities by minimizing typical losses~\eqref{eq:eq_ae}  acting only on a database of examples (approximately) in $\Sigma$. This will lead us to use images that are outside $\Sigma$, such as pure noise images, for the purposes of training. This is, in itself, quite novel in the literature.
 
 In the following, we propose a regularization of the training loss based on the definition of $\Psi_P$. We will show experimentally that the proposed regularization leads to improved projections in the context of GPGD for inverse problems. Note that in practice, with stochastic orthogonal regularization, we can have control of $\Psi_P \leq c < 1 $, with e.g. $c \approx 0.2$ in the MNIST case (Figure~\ref{fig:orth_values}), which fits the hypothesis of Theorem~\ref{theorem:eta_L} for $\Psi_P$. Controlling $\Phi_P$ (or showing that $\Phi_P$ is naturally controlled with our orthogonal regularization) during the training is left as an open question. 
 
 \subsection{Stochastic orthogonal regularization of deep projective priors}\label{section:section_ortho}
 
 We now present our stochastic orthogonal regularization (SOR) loss, used in the context of training a deep prior (e.g. autoencoders or denoising networks) on a database $X$. We suppose that the images of our dataset are in $[0,1]^n$. 
 
 We introduce the stochastic orthogonal regularization $\mathcal{R}$ as:
 \begin{equation}
 	\mathcal{R}(P):=\mathbb{E}_{z\sim\mathcal{U}([0,1]^n)} \left[\psi_P(z)\right]=\mathbb{E}_{z\sim\mathcal{U}([0,1]^n)}\left[\frac{|\langle P(z),z-P(z)\rangle|}{\norm{P(z)}{2}\norm{z-P(z)}{2}} \right].
 \end{equation}	
 
 We introduce the corresponding regularized losses for the training of deep projective priors based on autoencoders (AEs) or plug-and-play (PnP) priors: 
 \begin{equation}\label{eq:loss}
 	\begin{split}
 		\mathcal{L}_{X,\lambda,\text{AE}}^{\mathrm{reg}}(P)&=\mathcal{L}_{X,\text{AE}}(P) + \lambda \mathcal{R}(P);	\\
 		\mathcal{L}_{X,\lambda,\text{PnP}}^{\mathrm{reg}}(P)&=\mathcal{L}_{X,\text{PnP}}(P) + \lambda \mathcal{R}(P);	\\		
 	\end{split}		
 \end{equation}
 
 where $\lambda\in\mathbb{R}^+$ is a regularization parameter and $\mathcal{L}_{X,\text{AE}/\text{PnP}}$ are the classical Mean Squared Error (MSE) loss introduced in~\eqref{eq:eq_ae} and calculated over the training dataset $X$. Although we present SOR in the case of autoencoders and PnP denoisers, it can obviously be applied to any deep projective prior.
 
 To compute $\mathbb{E}_{z\sim\mathcal{U}([0,1]^n)} \left[\psi_P(z)\right]$, we use an empirical mean $\frac{1}{s}\sum_{i=1}^s \psi_P(z_i)$ (where $z_i \sim \mathcal{U}([0,1]^n)$) calculated using $s$ random points. Hence, we approximate the gradient of $\mathcal{R}$ with a stochastic gradient, making this orthogonal regularization \emph{stochastic}. The choice of a uniform distribution is justified by our desire to have a global orthogonal projection on the set containing all images. 
 In fact, apart from the close surroundings of $\Sigma$ where the MSE guarantees a good approximation of the orthogonal projection (as they minimize the same squared error), no constraint is applied over the rest of the space. Our choice assumes that images are represented in $[0,1]^n$ as common datasets for inverse problems are normalized to this range. Still, it is an open question to determine what is the optimal distribution for SOR. 
 
 To compute SOR, we generate at each epoch, and for each batch $X_s$ of size $s$ of the training dataset, another batch of the same size of random images $z_i \sim \mathcal{U}([0,1]^n)$. Suppose we train a projection $P_\theta$ parametrized by $\theta$, the SGD iterations can be written as: 
 \begin{equation}
 	\begin{split}
 		\theta_{k+1} &= \theta_k - \tau G(\theta_k) \\ 
 		G(\theta) &= \nabla_{\theta, X_s} \mathcal{L}_{X,\text{AE/PnP}}(P_
 		\theta) + \frac{1}{s}\sum_{i=1}^s \nabla_\theta \psi_{P_\theta}(z_i),
 	\end{split}
 \end{equation}
 where $\tau$ is the learning rate and $\nabla_{\theta,X_s} \mathcal{L}_{X,\text{AE/PnP}}(P_\theta)$ is the empirical mean of the gradient of the MSE loss over $X_s$. We consequently verify that $\mathbb{E} [G(\theta_k)] = \nabla_\theta 	\mathcal{L}_{X,\lambda,\text{AE/PnP}}^{\mathrm{reg}}(P_{\theta_k})$. Learning with SOR thus has the same convergence guarantees as conventional SGD learning schemes. In practice, we use the Adam algorithm for training the neural networks of the experiments section.
 
 With stochastic orthogonal regularization, we expect a better control of the restricted Lipschitz constant of deep projective priors. Note that, according to our framework, any improvement of the constant $\beta$ will result in improved identifiability and convergence speed guarantees. This is an improvement to frameworks that require operators modeled by neural networks to be contractive (global Lipshitz constant bounded by $1$).   However, the SOR term might degrade the MSE performance of the projective prior. We will observe experimentally that it is possible to improve the orthogonality of $P$ while maintaining similar MSE in the experiments. In the next section, we implement SOR for deep projective priors based on autoencoders and on plug-and-play. We use the resulting regularized projections for solving imaging inverse problems with GPGD.

	\section{Experimental results}\label{section:exp}

In this Section, we present results that demonstrate the improved convergence speed and robustness of GPGD for inverse problems when used with Deep Projective Priors trained with Stochastic Orthogonal Regularization (as predicted by improved restricted Lipschitz constants). We recall that the GPGD algorithm is given in Equation~\eqref{eq:pgd}, where the projection $P_\Sigma$ is performed with a DPP. We consider two types of DPP with two different databases: autoencoders with MNIST and plug-and-play DPP with a dataset of CelebA (human faces) images. We consider the following inverse problems

\begin{itemize}
	
	\item Image super-resolution: this consists in estimating a high-resolution image from a low-resolution image (see e.g.~\cite{yue2016image,wang2020deep} for an overview of methods for image-super resolution, we will focus on comparisons to deterministic PnP methods in the following). In equation \eqref{eq:eq_pb_inverse}, the vector $y\in \mathbb{R}^m$ is the subsampled image and the linear operator $\mathbf{A}\in \mathbb{R}^{m\times n}$ represents a low pass filtering followed by a subsampling, i.e. $\mathbf{A}=\mathbf{S}\mathbf{F}$ where $\mathbf{S}$ is the sub-sampling matrix by a given factor. The matrix $\mathbf{F}$ models the convolution with a Gaussian blur.

	\item Missing pixel inpainting: this consists in recovering an image from an observation where a certain proportion of randomly chosen pixels are missing (ie set to 0). The operator $\mathbf{A}$ from (\ref{eq:eq_pb_inverse}) is a matrix that selects the pixels to be deleted. We define the ``ratio'' as the proportion of pixels that will be deleted from the whole image. Note that this experiment is mainly used for testing our methodology (as we can easily control the conditioning of $\mathbf{A}$). However it could be linked with classical inverse problems such as impulse noise removal~\cite{chan2004iterative}.
	
	\item For CelebA images, we consider a deblurring task where $\mathbf{A}$ models a convolution with a Gaussian blur with kernel of size $5 \times 5$.
\end{itemize}
The main goal of this section is to compare the results obtained via a GPGD with and without SOR.  In the context of plug-and-play, GPGD without regularization of the denoiser can be interpreted as the baseline plug-and-play proximal gradient descent algorithm. For the particular case of the denoiser DPP on CelebA, we add the \textit{Regularization by denoising} (RED  \cite{romano2017little}) algorithm that is out of the theoretical scope of this paper, but still uses a projective prior. Hence, RED results are only to be used to help situate our results within the broader plug-and-play literature.

\textbf{Metrics:} To assess the results, we compare the Peak-Signal-to-Noise-Ratio (PSNR) of recovered images, \textit{i.e.} the solution of the inverse problem given by GPGD. We quantify the convergence speed by considering the first iteration $i$ such that $\frac{\|x_i-x^*\|_2}{\|x^*\|_2}\leq 0.01$ (0.005 for CelebA), where $x^*$ is the image given by the best of the 150 iterations of the GPGD \textit{i.e.} the one with the best PSNR.

The code to reproduce experiments is provided in \cite{joundi_code2024}.

\subsection{SOR of autoencoders DPP on MNIST}\label{section:ae}

As a proof of concept, we first solve in this subsection an inverse problem via GPGDs that use convolutional autoencoders regularized with SOR with different regularization parameters $\lambda$. They were trained during 2000 epochs with batches of size 64. Hence, at each iteration, a batch of size 64 of random images in $[0,1]^{28\times 28}$ is generated to train SOR. The size of the dataset is 30000. We put the precise details of the architectures in the Appendix \ref{section:architecture_reseaux}. In the appendix \ref{section:results_lambda}, we verify that SOR reduces the orthogonality criterion over different types of images. 

In the following inverse problems, the noise $e$ from the model (\ref{eq:eq_pb_inverse}) has a variance of $\sigma=0.02$

\noindent \textbf{Image super-resolution}	We perform super-resolution from images down-sampled by a factor of $2$. Figure \ref{fig:convergence_speed} represents a graph of PSNR recovery and convergence speed as a function of the weight $\lambda$. We complement the graph with visual results of the recovered image for 3 different $\lambda$: 0, 0.4, 5. We observe in the graph that the convergence is faster as $\lambda$ increases, while the PSNR remains stable until $\lambda \approx 0.4$. Thus, an autoencoder trained with SOR significantly accelerates the convergence of PGD algorithm with recovered images of equal quality.

\begin{figure}[t]
	\centering
	\begin{subfigure}{0.4\textwidth}
		\includegraphics[width=1\textwidth]{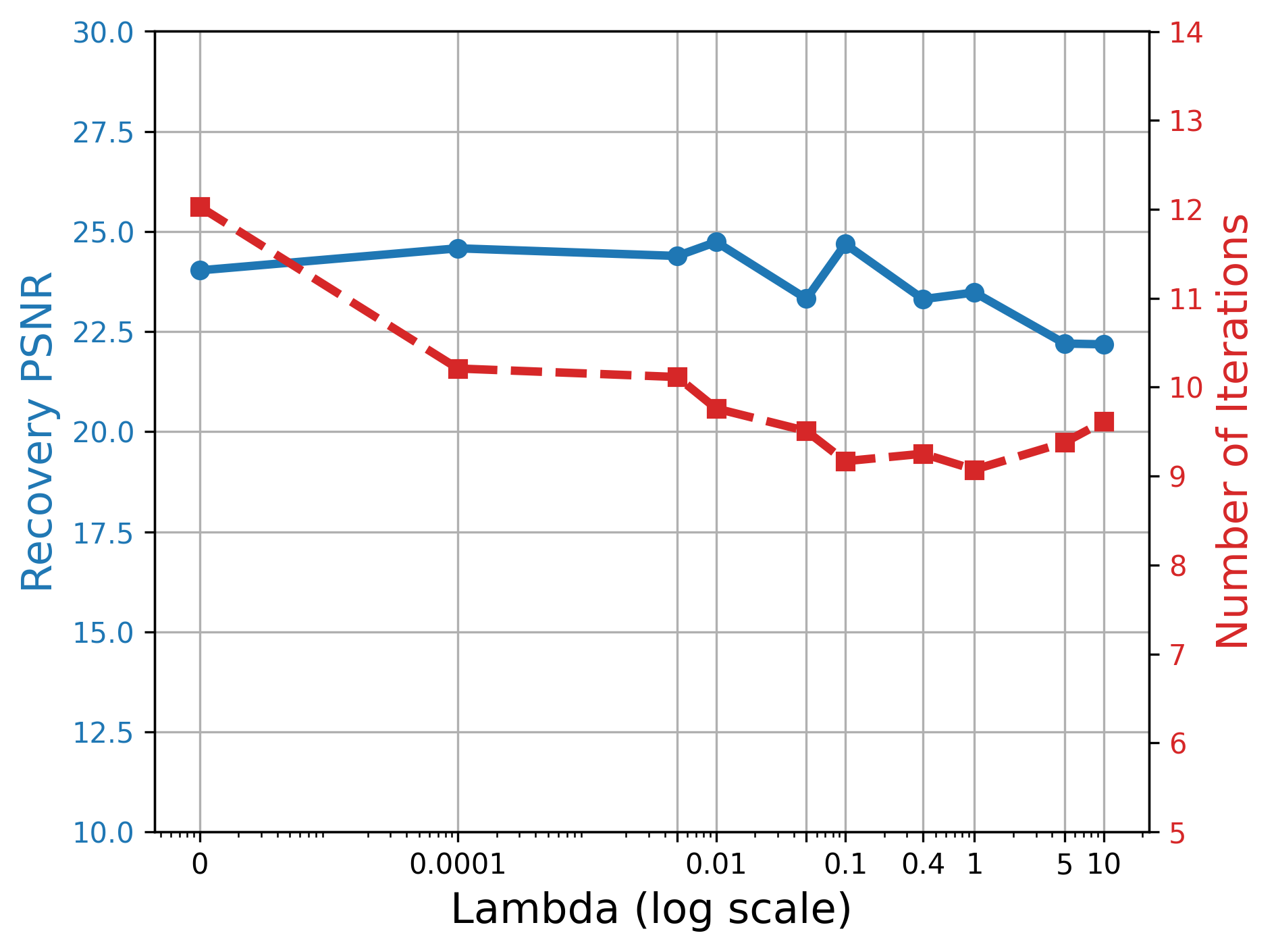}
		\caption*{Recovery PSNR and convergence speed as a function of the weight $\lambda$}\label{graph_supresmnist}
	\end{subfigure}
	\hspace{0.1cm}
	\begin{subfigure}{0.45\textwidth}
		\includegraphics[width=1\textwidth]{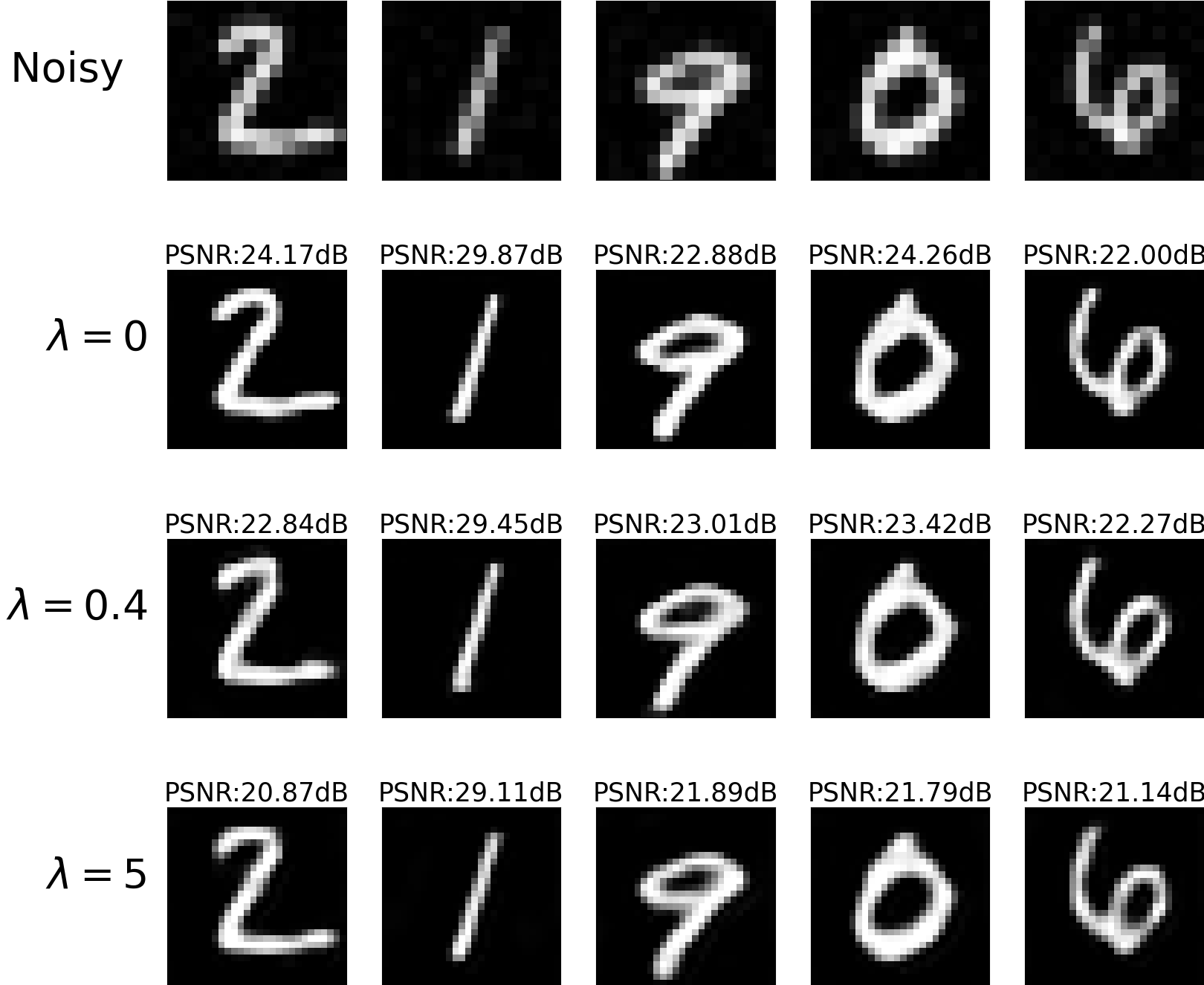}
		\caption*{Recovery of subsampled MNIST images through different PGDs.}
	\end{subfigure}
	\caption{Recovery of MNIST images from a super-resolution inverse problem with different PGDs using autoencoders weighted differently. The more lambda increases, the quicker the convergence, and this is without a significant loss of recovery PSNR. }\label{fig:convergence_speed}
\end{figure}

\noindent \textbf{Missing pixel inpainting}
For the inpainting problem, we test a wide range of ratios of missing pixels and compare the recovery of PGDs using a classical autoencoder ($\lambda=0$) and a regularized autoencoder ($\lambda=0.4$) that gives the best tradeoff between the MSE and SOR. Figure \ref{fig:Mask_random_MNIST} shows a graph of the PSNR as a function of the ratio and a visual result of an image recovery from MNIST.

\begin{figure}[!h]
	\centering
	\begin{subfigure}{0.4\textwidth}
		\includegraphics[width=1\textwidth]{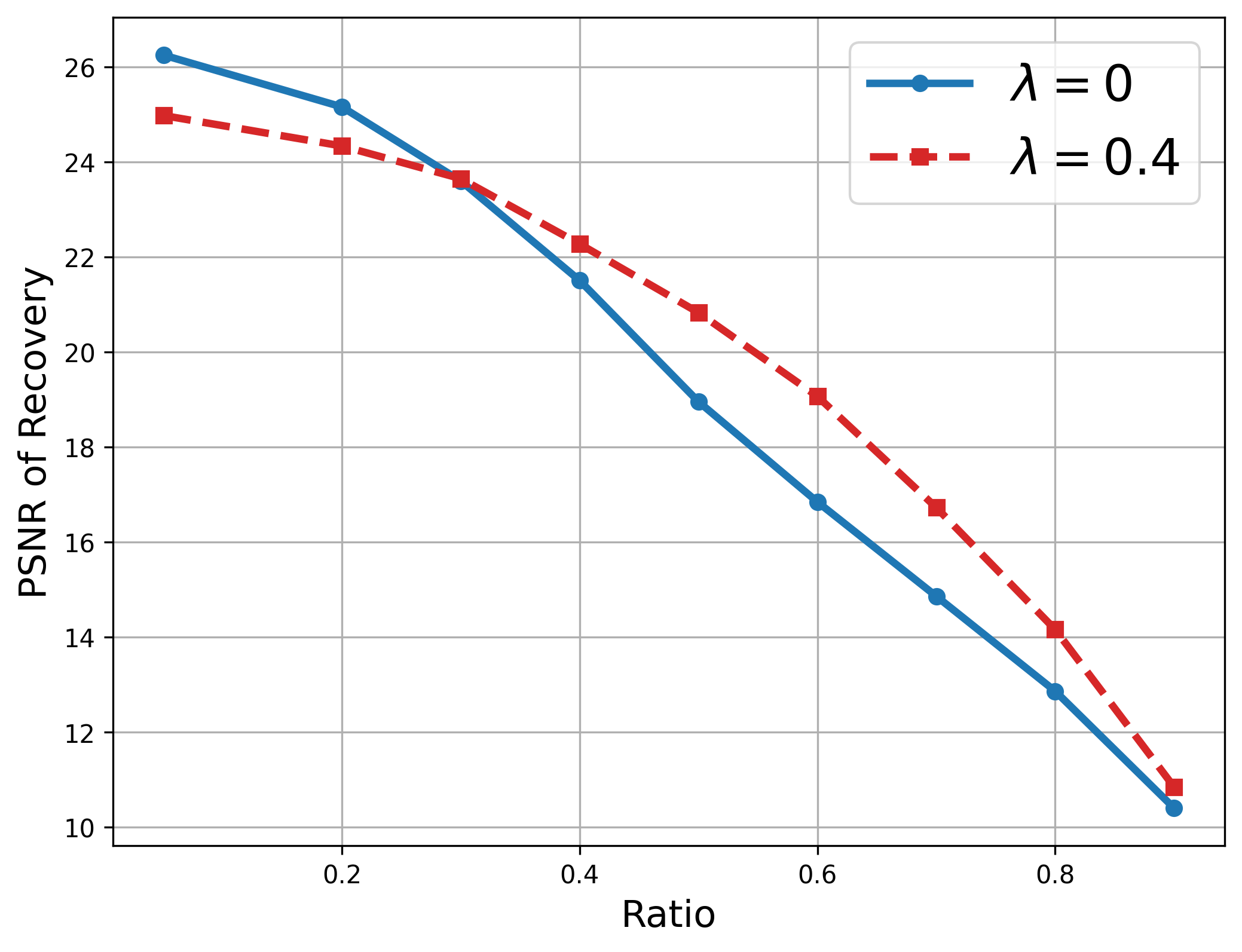}
		\caption*{Evolution of the mean PSNR of recovery as a function of the ratio of missing pixels}\label{fig:graph_MNIST_mask_random}
	\end{subfigure}
	\hspace{0.33cm}
	\begin{subfigure}{0.1565\textwidth}
		\includegraphics[width=1\textwidth]{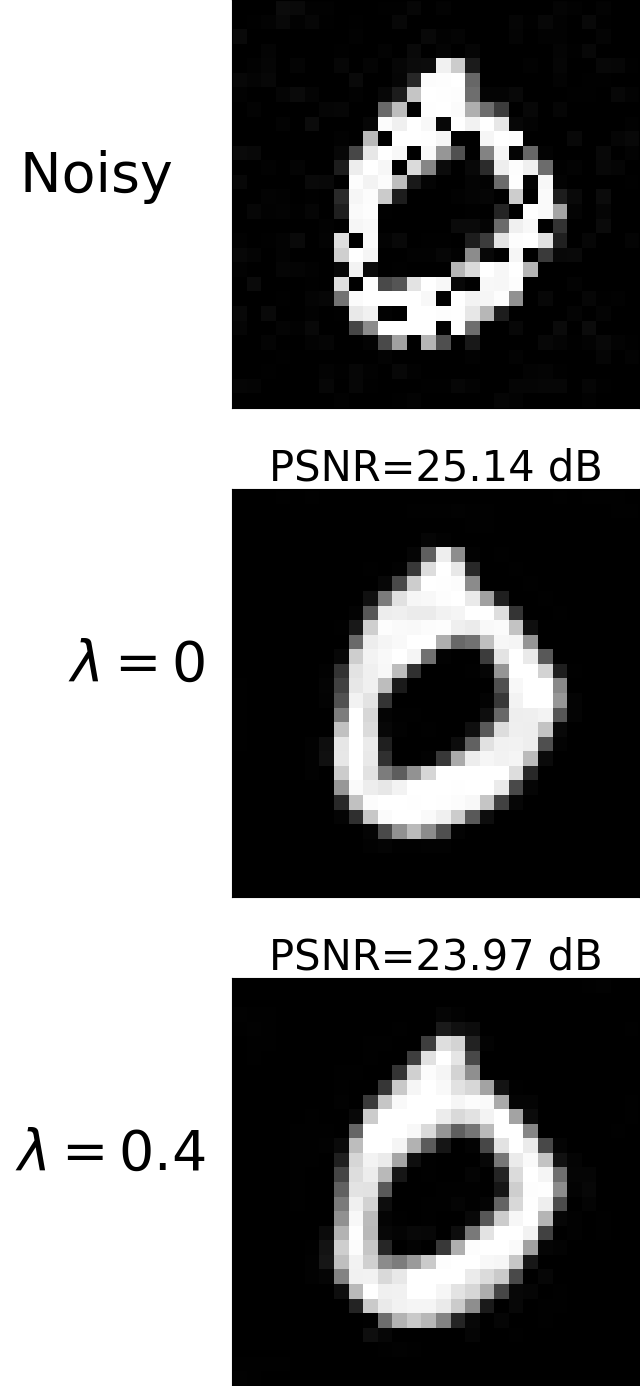}
		\caption*{Ratio: ~~~~0.2~~~~~} 
	\end{subfigure}
	\hspace{0.1cm}
	\begin{subfigure}{0.1\textwidth}
		\includegraphics[width=1\textwidth]{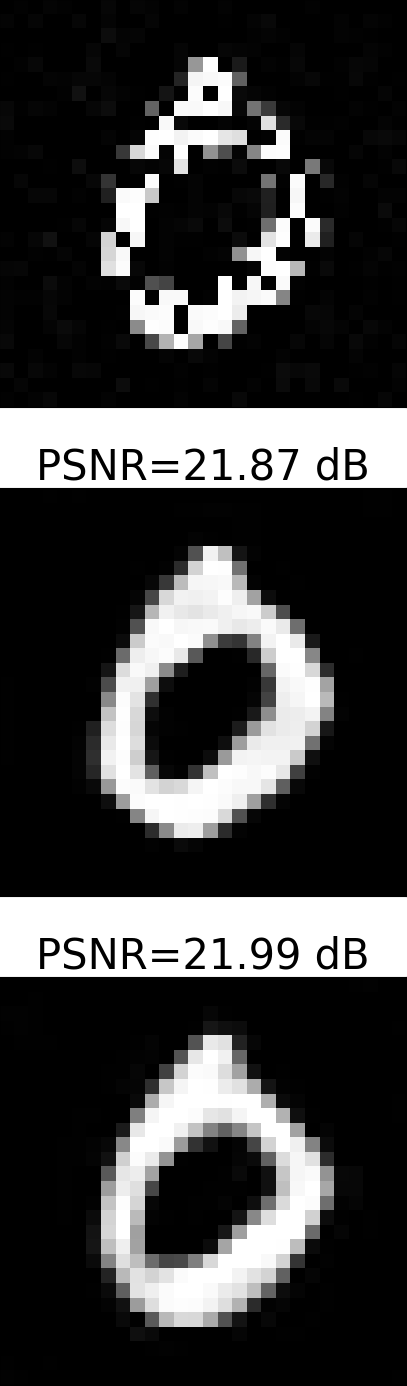}
		\caption*{0.4}
	\end{subfigure}
	\hspace{0.1cm}
	\begin{subfigure}{0.1\textwidth}
		\includegraphics[width=1\textwidth]{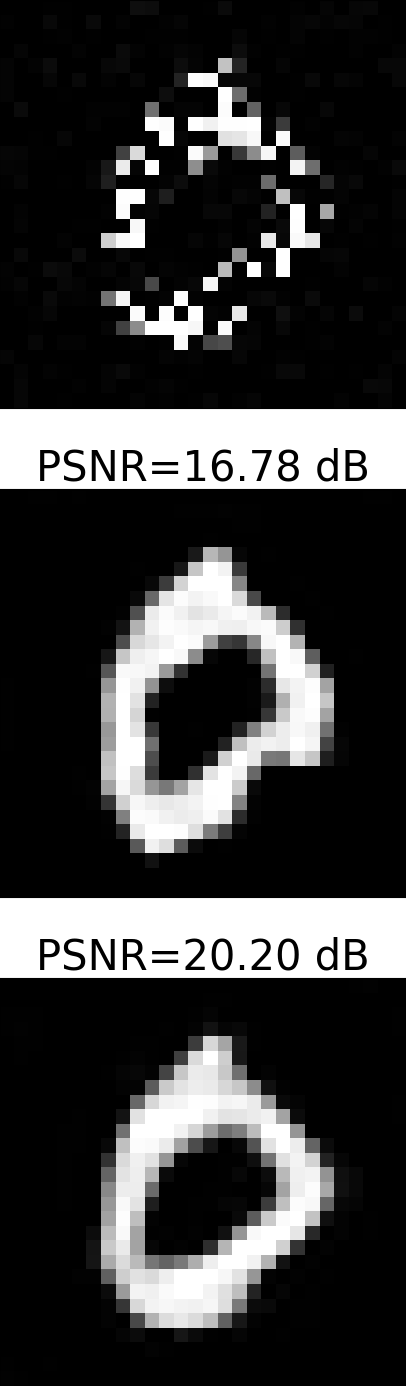}
		\caption*{0.6}
	\end{subfigure}
	\hspace{0.1cm}
	\begin{subfigure}{0.1\textwidth}
		\includegraphics[width=1\textwidth]{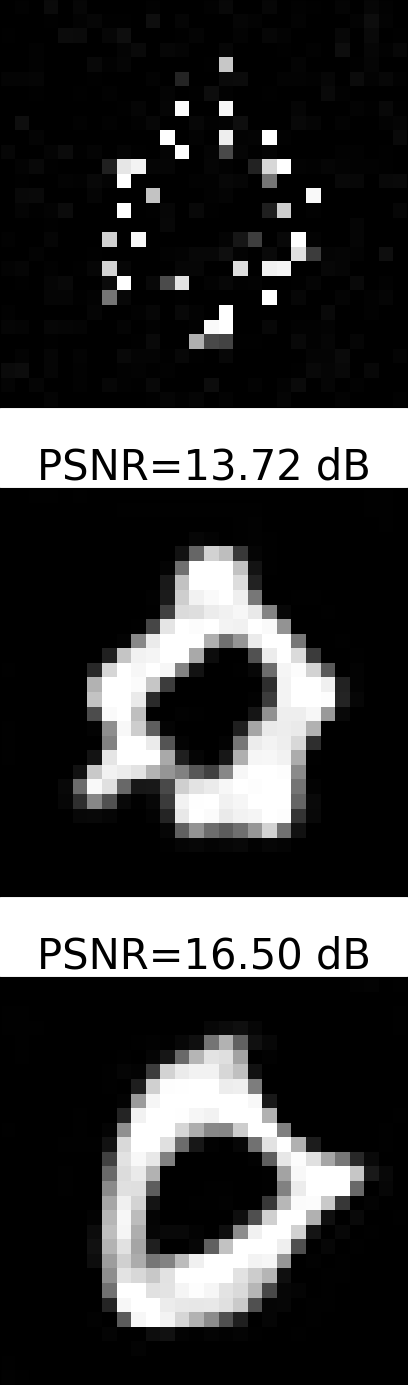}
		\caption*{0.8}
	\end{subfigure}	
	
	\caption{Graph and visual results of the evolution of the recovered MNIST image with the increase of the inpainting ratio for the different autoencoders. When using a regularized autoencoder, the recovery is more robust to the deterioration of the measurement matrix. }\label{fig:Mask_random_MNIST}
\end{figure}

The graph \ref{fig:Mask_random_MNIST} shows that the PSNR of the recovered images deteriorates when increasing the ratio, which is an expected behavior.  The basic autoencoder performs better when considering low inpainting ratios. In fact, it is explained by the observed images being near $\Sigma$. Then, autoencoding performance determines the quality of the reconstruction.   
However, it is outperformed by the regularized autoencoders in ratios higher than 0.4. This is in particular shown in Figure~\ref{fig:Mask_random_MNIST}: the regularized autoencoder keeps an acceptable shape of the 0 compared to the basic one that deteriorates more significantly. In other words, SOR loss for training the prior improved solutions in challenging circumstances (with many missing pixels) both numerically and visually. We attribute this to a lower restricted Lipschitz constant $\beta$ gained by the SOR. Indeed, when $\mathbf{A}$ is deteriorated, its RIC $\delta$ increases, hence a greater convergence rate $r=\delta \beta$ of GPGD (\ref{th:gen_convevergence}). A lower restricted Lipschitz constant $\beta$ compensates this increase of $r$ and also improves the stability constant (that is increasing with $\beta$, see Theorem~\ref{th:gen_convevergence}).

	\subsection{SOR of Plug and Play DPP on CelebA images}\label{section:pnp-celebA}

We now apply our SOR loss to another deep projective prior, deep denoising neural networks on a more realistic dataset. We train a DRUNET denoiser \cite{zhang2021plug}, a U-net combined with skip connections on a dataset of size 10000 of CelebA images \cite{liu2015faceattributes}. The images represent RGB face images of size 256x256. We considered for these two denoisers without SOR ($\lambda=0$) and with SOR ($\lambda=0.1$). We included in the appendix \ref{section:architecture_reseaux} details on the training and the architectures of the DRUNETs that were used. 

After the training, we solve three different inverse problems: inpainting, super-resolution and deblurring. We consider three image restoration algorithms: the classical unconstrained GPGD (noted in the Figures PGD), SOR, and Regularization by denoising. The main baseline here is PGD.
PGD and RED use the same non-regularized DRUNET, whereas SOR uses the regularized one.

For each inverse problem, we set different noise levels with the parameter $\sigma$. We consider for our experiments a test set of 50 images from CelebA. Additionally, a second metric is computed: Structural Similarity Index Measure (SSIM). Tables \ref{table:Inpainting_0.4}, \ref{table:Inpainting_0.6}, \ref{table:SUPRES_2}, \ref{table:SUPRES_3} and \ref{table:GBlur} present our results for the considered inverse problems. Visual results are provided to illustrate these results. Note that additional experiments are available in the appendix.

\subsubsection*{Inpainting}
\begin{table}[!h]
	\centering
	\begin{minipage}{0.48\textwidth}
		\centering
		\small
		\begin{tabular}{ccccc}
			\toprule
			$\sigma$           & Method & PSNR$\uparrow$  & SSIM$\uparrow$ & Iterations$\downarrow$\\
			\midrule
			\multirow{3}{*}{0}    & PGD    & 33.35 & 0.95 & 27.62          \\
			& SOR    & 35.39 & 0.95 & \textbf{11.58 }       \\
			& RED    & \textbf{36.66} & \textbf{0.96} & 38.64         \\
			\hline
			\multirow{3}{*}{0.02} & PGD    & \textbf{35.63} & \textbf{0.96} & 27.68        \\
			& SOR    & 34.93 & 0.95 & \textbf{10.42 }         \\
			& RED    & 34.64 & 0.92 & 30.96 \\
			\bottomrule      
		\end{tabular}
	\caption{Metrics of an inpainting ratio of 0.4: SOR is the fastest converging algorithm while providing competitive recovery results compared to the two other methods.}\label{table:Inpainting_0.4}
\end{minipage}
\hfill
\begin{minipage}{0.48\textwidth}
	\centering
	\small
	\begin{tabular}{ccccc}
		\toprule
		$\sigma$           & Method & PSNR$\uparrow$  & SSIM$\uparrow$ & Iterations$\downarrow$\\
		\midrule
		\multirow{3}{*}{0}    & PGD    & 19.78 & 0.56 & \textbf{4.92}         \\
		& SOR    & \textbf{32.46} & \textbf{0.93} & 18.64        \\
		& RED    & 29.58 & 0.83 & 37.18        \\
		\hline
		\multirow{3}{*}{0.02} & PGD    & 20.14 & 0.58 & \textbf{5.72}         \\
		& SOR    & \textbf{32.33} &\textbf{ 0.92} & 17.88         \\
		& RED    & 28.81 & 0.81 & 36.36\\
		\bottomrule  
	\end{tabular}
\caption{Metrics of an inpainting ratio of 0.6: SOR gives the best recovery results while still being quicker in both noiseless and noisy cases.}\label{table:Inpainting_0.6}
\end{minipage}

\end{table}

For the inpainting inverse problem. Tables \ref{table:Inpainting_0.4} and \ref{table:Inpainting_0.6} show that the more complex the problem is, the more SOR outperforms the two other methods. While it is slightly below RED in an inpainting of ratio 0.4 and $\sigma=0$, the recovery is improved when moving from an inpainting ratio of 0.4 to 0.6 or from a noise level of 0 to 0.02. This is what we observed for the MNIST case. We provide an additional inpainting ratio of 0.7 in the appendix, showing in particular that SOR is still outperforming in every metric.

\begin{figure}[!h]
\centering
\begin{minipage}{\textwidth}
\centering
\begin{subfigure}{0.23\textwidth}
	\includegraphics[width=1\textwidth]{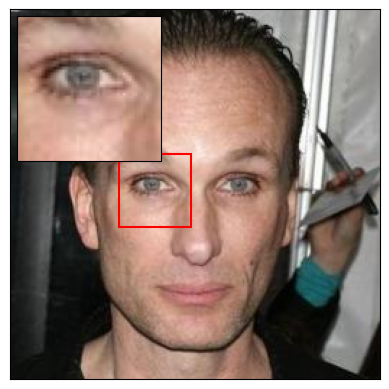}
	\caption{Original}
\end{subfigure}
\hspace{0.02cm}
\begin{subfigure}{0.23\textwidth}
	\includegraphics[width=1\textwidth]{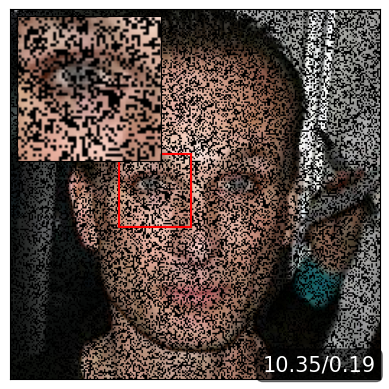}
	\caption{Observed}
\end{subfigure}
\hspace{0.02cm}
\begin{subfigure}{0.23\textwidth}
	\includegraphics[width=1\textwidth]{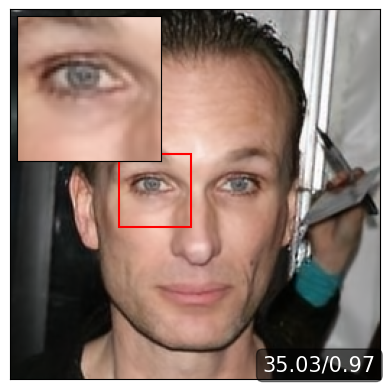}
	\caption{PGD}
\end{subfigure}

\begin{subfigure}{0.23\textwidth}
	\includegraphics[width=1\textwidth]{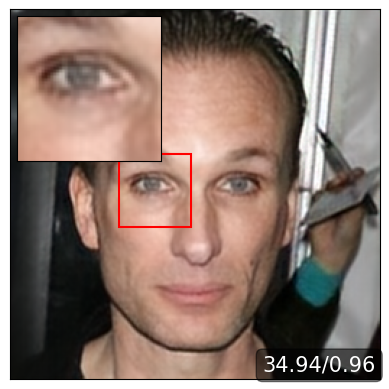}
	\caption{SOR}
\end{subfigure}
\hspace{0.02cm}
\begin{subfigure}{0.23\textwidth}
	\includegraphics[width=1\textwidth]{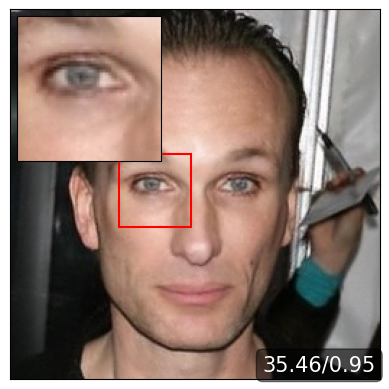}
	\caption{RED}
\end{subfigure}
\caption{Inpainting with the three considered methods of an image with an inpainting ratio of 0.4 and noise level of 0. RED recovers better the original image but is slightly smoothing the image compared to the GPGD methods}\label{fig:inpainting_0.4}
\end{minipage}
\vfill
\begin{minipage}{\textwidth}
\centering
\begin{subfigure}{0.23\textwidth}
	\includegraphics[width=1\textwidth]{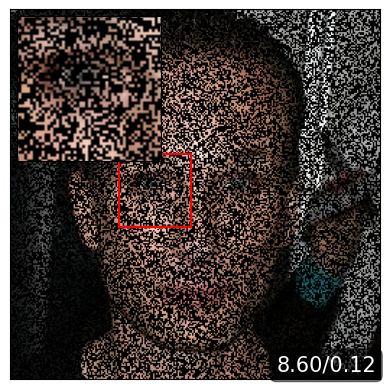}
	\caption{Observed}
\end{subfigure}
\hspace{0.02cm}
\begin{subfigure}{0.23\textwidth}
	\includegraphics[width=1\textwidth]{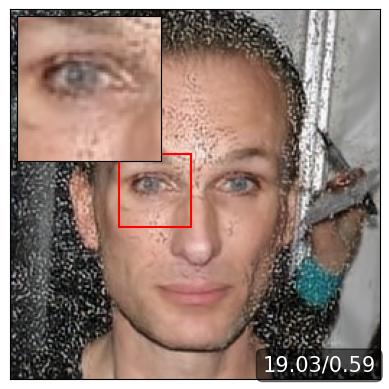}
	\caption{PGD}
\end{subfigure}
\hspace{0.02cm}
\begin{subfigure}{0.23\textwidth}
	\includegraphics[width=1\textwidth]{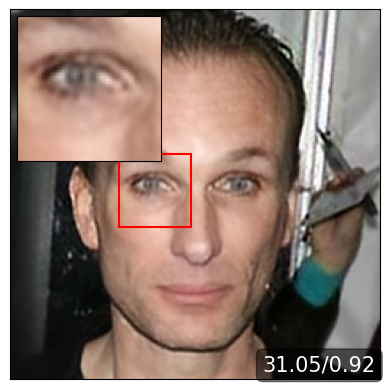}
	\caption{SOR}
\end{subfigure}
\hspace{0.02cm}
\begin{subfigure}{0.23\textwidth}
	\includegraphics[width=1\textwidth]{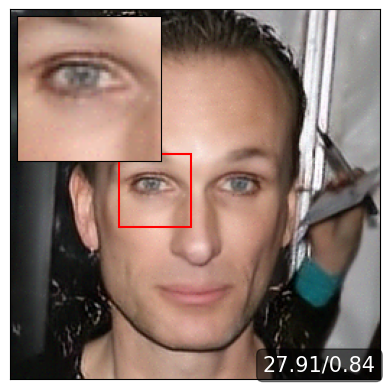}
	\caption{RED}
\end{subfigure}
\caption{Inpainting an image with an inpainting ratio of 0.6 and noise level of 0. While PGD produces many artifacts, SOR and RED correctly recover the original image with the same performance. Visually, SOR produces more realistic and less smoothed images compared to RED.}\label{fig:inpainting_0.6}
\end{minipage}
\end{figure}

Visually, Figures \ref{fig:inpainting_0.4} and \ref{fig:inpainting_0.6} show that while PGD creates lot of artifacts when going from an inpainting ratio of 0.4 to 0.6, SOR keeps a clean image. In addition to that, the obtained image preserves the textures compared to RED that tends to average the image. 

Furthermore, another important result can be seen from the Table \ref{table:Inpainting_0.4} and \ref{table:Inpainting_0.6}. When comparing the convergence speed of both methods, we observe clearly that, when the PSNRs are similar, SOR is, once again, clearly quicker than both PGD and RED ($2.5$ times faster than both methods for a ratio $0.4$). This confirms two things: first, GPGD algorithms converge faster compared to RED, which has been discussed in \cite{traonmilin2024towards}. Second, our regularization that aims to lower the restricted Lipschitz constant has, in fact, accelerated the convergence speed compared to PGD that is not regularized.

 \subsubsection*{Super-resolution}
\begin{table}[!h]
	\centering
	\begin{minipage}{0.48\textwidth}
		\centering
		\small
		\begin{tabular}{ccccc}
			\toprule
			$\sigma$          & Method & PSNR$\uparrow$  & SSIM$\uparrow$ & Iterations$\downarrow$\\
			\midrule
			\multirow{3}{*}{0.02}    & PGD    & \textbf{34.22} & \textbf{0.93} & 4.44         \\
			& SOR    & 32.87 & 0.91 & \textbf{3.0}          \\
			& RED    & 33.70 & 0.92 & 12.6         \\
			\hline
			\multirow{3}{*}{0.05} & PGD    & \textbf{31.45} & \textbf{0.88} & 8.18          \\
			& SOR    & 30.82 & 0.85 & \textbf{2.16}          \\
			& RED    & 30.90 & 0.86 & 15.66\\
			\bottomrule             
		\end{tabular}
		\caption{Metrics of a super-resolution of factor 2: PGD gives the best recovery results but still, SOR is the fastest converging algorithm.}\label{table:SUPRES_2}
	\end{minipage}
	\hfill
	\begin{minipage}{0.48\textwidth}
		\centering
		\small
		\begin{tabular}{ccccc}
			\toprule
			$\sigma$           & Method & PSNR$\uparrow$  & SSIM$\uparrow$ & Iterations$\downarrow$\\
			\midrule
			\multirow{3}{*}{0.02}    & PGD    & 25.71 & 0.80 & 28.96        \\
			& SOR    & \textbf{29.67} & \textbf{0.86} & \textbf{8.26}         \\
			& RED    & 24.77 & 0.75 & 29.3         \\
			\hline
			\multirow{3}{*}{0.05} & PGD    & 25.34 & \textbf{0.79} & 29.38         \\
			& SOR    & \textbf{28.20} & \textbf{0.79} & \textbf{7.58}         \\
			& RED    & 24.19 & 0.72 & 27.18 \\
			\bottomrule                  
		\end{tabular}
		\caption{Metrics of a super-resolution of factor 3: SOR gives the best recovery images while being the fastest among all the methods.}\label{table:SUPRES_3}
	\end{minipage}
\end{table}

The same observations can be made for the super-resolution inverse problem. Whereas SOR is slightly weaker in recovering images when the super-resolution factor is 2, it outperforms both PGD and RED when increasing this factor to 3. This is explained by the visual results in Figure \ref{fig:SUPRES_2} and \ref{fig:SUPRES_3}. We see that SOR creates less ringing effect compared to PGD and RED as seen in Figure \ref{fig:SUPRES_3}. This shows that SOR adds clearly a robustness to aliasing which can be modeled as a structured noise (again through the stability constant given by Theorem \ref{th:gen_convevergence}). We provide in the appendix experiments on MRI images that account for the same observations.

The increased convergence speed of SOR is still valid for the super-resolution. The particularity of this case is that, for a super resolution of factor 2, SOR can converge within very few iterations while obtaining competitive results compared to PGD and RED.

\begin{figure}[!h]
	\centering
	\begin{subfigure}{0.23\textwidth}
		\includegraphics[width=1\textwidth]{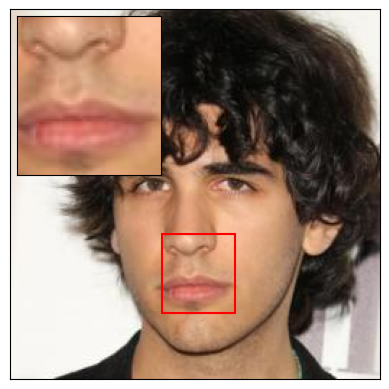}
		\caption{Original}
	\end{subfigure}
	\hspace{0.02cm}
	\begin{subfigure}{0.23\textwidth}
		\includegraphics[width=1\textwidth]{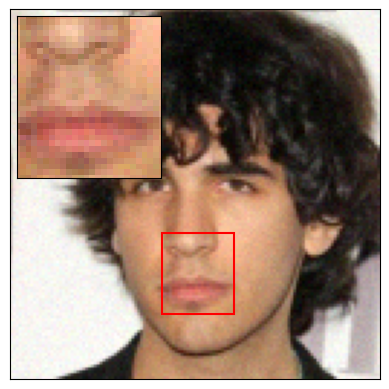}
		\caption{Observed}
	\end{subfigure}
	\hspace{0.02cm}
	\begin{subfigure}{0.23\textwidth}
		\includegraphics[width=1\textwidth]{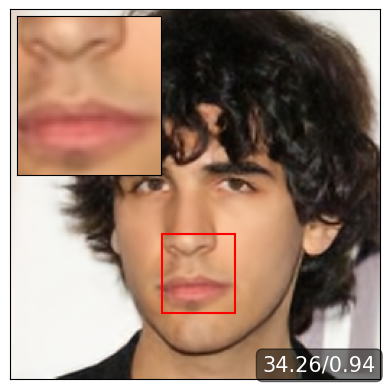}
		\caption{PGD}
	\end{subfigure}
	
	\begin{subfigure}{0.23\textwidth}
		\includegraphics[width=1\textwidth]{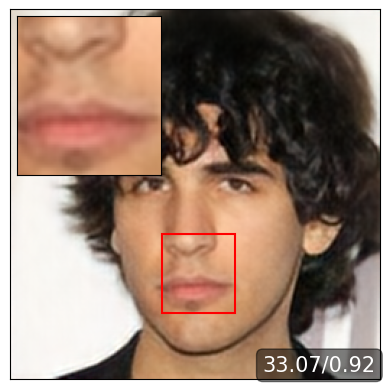}
		\caption{SOR}
	\end{subfigure}
	\hspace{0.02cm}
	\begin{subfigure}{0.23\textwidth}
		\includegraphics[width=1\textwidth]{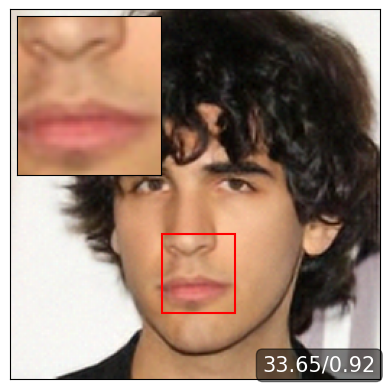}
		\caption{RED}
	\end{subfigure}
	\caption{Super-resolution of an image with a factor of 2 and a noise level of 0.02. All methods correctly recover the original image, with PGD being slightly the best.}\label{fig:SUPRES_2}
\end{figure}
\begin{figure}[!h]
	\centering
	\begin{subfigure}{0.23\textwidth}
		\includegraphics[width=1\textwidth]{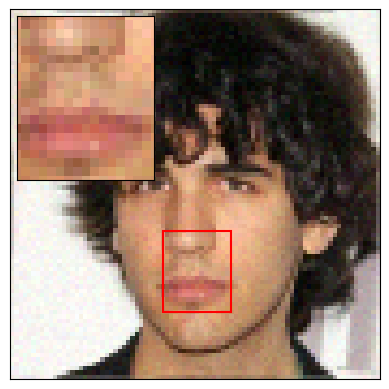}
		\caption{Observed}
	\end{subfigure}
	\hspace{0.02cm}
	\begin{subfigure}{0.23\textwidth}
		\includegraphics[width=1\textwidth]{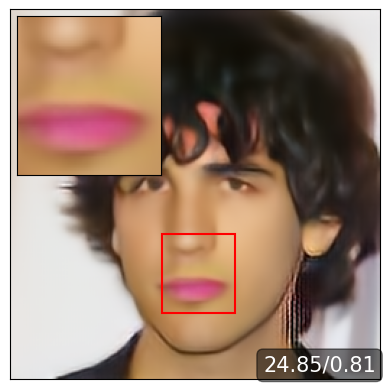}
		\caption{PGD}
	\end{subfigure}
	\hspace{0.02cm}
	\begin{subfigure}{0.23\textwidth}
		\includegraphics[width=1\textwidth]{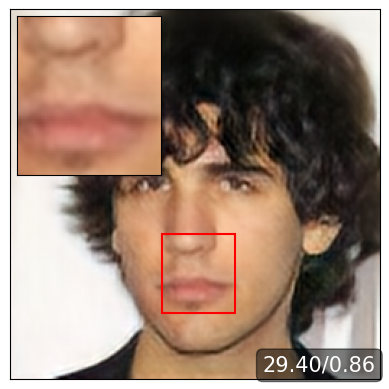}
		\caption{SOR}
	\end{subfigure}
	\hspace{0.02cm}
	\begin{subfigure}{0.23\textwidth}
		\includegraphics[width=1\textwidth]{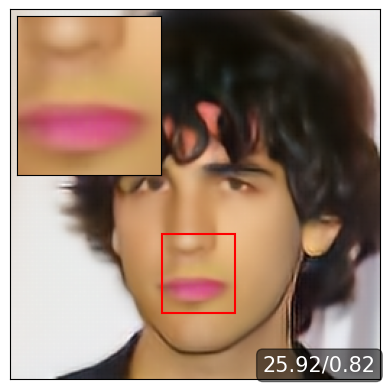}
		\caption{RED}
	\end{subfigure}
	\caption{Super-resolution of an image with a factor of 3 and a noise level of 0.02. SOR is the best in recovering the original image as it creates less color artifacts and does not smooth the image, unlike the two other methods.}\label{fig:SUPRES_3}
\end{figure}
\newpage

 \subsubsection*{Deblurring}
\begin{table}[!h]
	\centering
	\begin{tabular}{ccccc}
		\toprule
		$\sigma$           & Method & PSNR$\uparrow$  & SSIM$\uparrow$ & Iterations$\downarrow$\\
		\midrule
		\multirow{3}{*}{0.02}    & PGD    & 29.21 & 0.83 & 20.7         \\
		& SOR    & 28.19 & 0.80 & \textbf{8.82}         \\
		& RED    & \textbf{29.79} & \textbf{0.85} & 39.56        \\
		\hline
		\multirow{3}{*}{0.05} & PGD    & 28.76 & \textbf{0.82} & 16.6         \\
		& SOR    & 27.93 & 0.79 & \textbf{7.9}         \\
		& RED    & \textbf{28.80} & 0.81 & 27.78  \\
		\bottomrule                       
	\end{tabular}
	\caption{Deblurring of an image with a 5$\times$5 Gaussian kernel and noise level of 0.02. The performances of all methods are close, while SOR is faster.} \label{table:GBlur}
	\end{table}
	
	Last, we compare the results of a deblurring inverse problem. We consider a 5$\times$5 Gaussian kernel. According to Table \ref{table:GBlur}, SOR is not as good as RED and PGD but is still faster. This lower recovery is explained by the fact that deblurring is less challenging compared to the inpainting and the super-resolution tasks. Yet, this problem points out the necessity to have a good balance between the denoising performances and SOR to obtain the optimal recovery and convergence speed.

	\begin{figure}
		\centering
		\begin{subfigure}{0.23\textwidth}
			\includegraphics[width=1\textwidth]{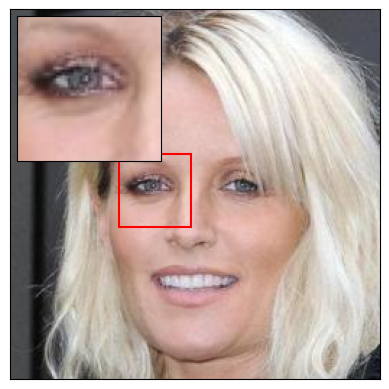}
			\caption{Original}
		\end{subfigure}
		\hspace{0.02cm}
		\begin{subfigure}{0.23\textwidth}
			\includegraphics[width=1\textwidth]{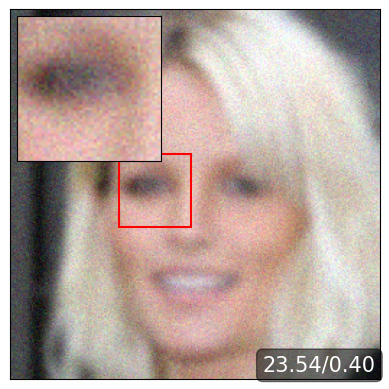}
			\caption{Observed}
		\end{subfigure}
		\hspace{0.02cm}
		\begin{subfigure}{0.23\textwidth}
			\includegraphics[width=1\textwidth]{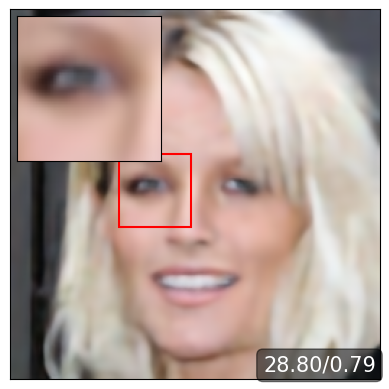}
			\caption{PGD}
		\end{subfigure}
		
		\begin{subfigure}{0.23\textwidth}
			\includegraphics[width=1\textwidth]{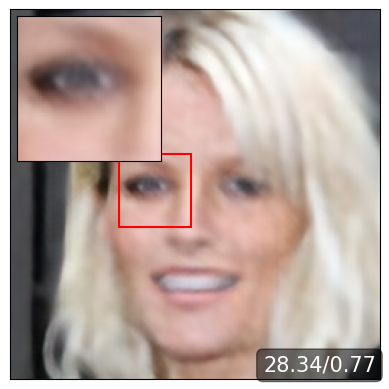}
			\caption{SOR}
		\end{subfigure}
		\hspace{0.02cm}
		\begin{subfigure}{0.23\textwidth}
			\includegraphics[width=1\textwidth]{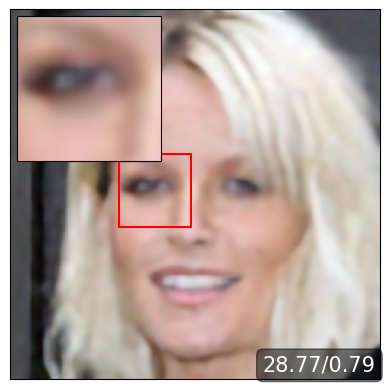}
			\caption{RED}
		\end{subfigure}	
		\caption{Deblurring of an image of a 5$\times$5  gaussian kernel and noise level of 0.05.  As deblurring inverse problem is better conditioned, RED and PGD have similar performance to SOR while SOR exhibits a $2$ to $3$ times speed up.}
	\end{figure}
	
	\vspace{0.5cm}
	In summary, we find that for every inverse problems, SOR
	loss speeds up the convergence of GPGD. For the inpainting and the super-resolution, it also allows a better identifiability of $\Sigma$, \textit{without $\Sigma$ ever being explicitly known}.
	\newpage

\section{Discussions and Limitations}\label{section:limitations}
While this work presents important and relevant results, it has some limitations that should be noted:
\begin{itemize}
	\item From a theoretical standpoint, we have supposed in the theorems that the orthogonal projection onto $\Sigma$ exists (\textit{i.e.} in $\mathbb{R}^n$, $\Sigma$ is a proximinal set). However, this is not generally the case. Generalizing our results to any set is still an open theoretical question.
	
	\item Experimentally, even if SOR proved its efficiency on different DPPs for multiple inverse problems, it does not fully use the decomposition exhibited in theorem \ref{theorem:eta_L}. However, the unused term $H$ in this theorem can be hard to exploit as it is closely related to the orthogonal projection over $\Sigma$ that is not accessible. 
	
	\item Assessing the reconstruction quality depends on the chosen metrics. We mainly use those that are derived from the $ L_2$-norm. They showed in particular that each methods behave differently in recovering the original images (by smoothing, or by preserving the textures but creating artifacts, etc.).  In future work, we could consider other norms and adapt SOR accordingly. 
	
\end{itemize}

\section{Conclusion}
In this paper, we have introduced the \emph{Stochastic Orthogonal Regularization}, a novel regularization that improves generalized projected gradient descent algorithms using deep projective priors. We theoretically justified the choice of SOR and showed experimentally that the resulting algorithms are both faster and more stable in the context of imaging inverse problems, with different choices of priors and datasets. 

Several open questions remain in our work. Firstly, we can investigate whether our choice of distribution of random images used for SOR is optimal. It would also be interesting to see if SOR may be used to improve the photo-realism of generative networks such as Variational Autoencoders or Diffusion Models. Indeed, the goal of such networks is to randomly sample ``realistic'' from a complex model of natural images. The main challenge would be to adapt the method to the non-deterministic setting of such priors.

\begin{MSCcodes}
68U10, 68T07, 90C26, 94A08
\end{MSCcodes}

\bibliographystyle{siamplain}
\bibliography{ref}
\newpage

	\appendix
\section{Proof of Theorem \ref{th:gen_convevergence}}
\begin{proof}

We have
\begin{equation}
	\begin{split}
		\|x_{n+1} -\gt\|_2^2 &= \| P_\Sigma(x_n)- \gamma \mathbf{A}^T(\mathbf{A}P_\Sigma(x_n)-y) -\gt \|_2^2\\
		&= \| P_\Sigma(x_n) - \gamma \mathbf{A}^T\mathbf{A}(P_\Sigma(x_n)-\gt) + \gamma \mathbf{A}^T e - \gt \|_2^2\\
		&=\| (\mathbf{I}-\gamma \mathbf{A}^T\mathbf{A})(P_\Sigma(x_n)-\gt) + \gamma \mathbf{A}^Te  \|_2^2\\
	\end{split}
\end{equation}
With the triangle inequality, the RIC $\delta = \delta( \mathbf{A}^T A)$ and the restricted $\beta$-Lipschitz condition of $P_\Sigma$, we have
\begin{equation}\label{eq:rec1}
	\begin{split}
		\|x_{n+1} -\gt\|_2 &\leq \| (\mathbf{I}-\gamma  \mathbf{A}^T\mathbf{A})(P_\Sigma(x_n)-\gt)\|_2 +\| \gamma \mathbf{A}^T e  \|_2\\
		&\leq \delta\| P_\Sigma(x_n)-\gt\|_2 +\| \gamma  \mathbf{A}^Te  \|_2\\
		&\leq \delta \beta\|x_n-\gt\|_2 +\| \gamma  \mathbf{A}^Te  \|_2.\\
	\end{split}
\end{equation}
We show by induction
\begin{equation}
	\begin{split}
		\|x_{n+1} -\gt\|_2 &\leq (\delta \beta)^{n+1}\|x_0-\gt\|_2+\left(\sum_{i=0}^{n}(\delta \beta)^i\right)\|\gamma \mathbf{A}^Te\|_2 \\
	\end{split}
\end{equation}
For $n=0$, this is given exactly by \eqref{eq:rec1}. For $n\geq 1$, suppose this inequality true at step $n-1$, with \eqref{eq:rec1}, we have 
\begin{equation}\label{eq:rec1}
	\begin{split}
		\|x_{n+1} -\gt\|_2 
		&\leq \delta \beta\|x_n-\gt\|_2 +\|\gamma  \mathbf{A}^Te  \|_2.\\
		&\leq (\delta \beta)^{n+1}\|x_0-\gt\|_2 +\delta\beta\left(\sum_{i=0}^{n-1}(\delta \beta)^i\right)\|\gamma  \mathbf{A}^Te\|_2 +\|\gamma  \mathbf{A}^Te  \|_2 \\
		&=(\delta \beta)^{n+1}\|x_0-\gt\|_2+\left(\sum_{i=0}^{n}(\delta \beta)^i\right)\|\gamma  \mathbf{A}^Te\|_2 .\\
	\end{split}
\end{equation}

Hence,when $\delta\beta <1$,

\begin{equation}
	\begin{split}
		\|x_{n} -\gt\|_2	& \leq (\delta \beta)^{n}\|x_0-\gt\|_2 + \frac{1}{1-\delta \beta}\|\gamma  \mathbf{A}^T e\|_2
	\end{split}
\end{equation}

\end{proof}

\section{Proof of Theorem \ref{theorem:general_rlc}}

\begin{proof}[Proof of theorem \ref{theorem:general_rlc}]
Let $z\in\bR^n$, $x\in\Sigma$. We denote by $\orth$ the orthogonal projection onto $\Sigma$ and $P$ a given projection onto $\Sigma$. We want to show that $P$ has a restricted Lipschitz constant $\beta$, \textit{i.e.} that there  is a $\beta$ such that $\norm{P(z)-x}{2}\leq \beta \norm{x-z}{2}$. 
By the triangular inequality and with the hypothesis, we have: 
\begin{equation}
	\begin{split}
		\norm{P(z)-x}{2} & \leq \norm{P(z)-\orth(z)}{2} + \norm{\orth(z)-x}{2}\\
		& \leq L\norm{z-\orth(z)}{2}+\norm{\orth(z)-x}{2}.\\
	\end{split}
\end{equation}

By definition of $\beta^\perp = \beta_\Sigma (P_\Sigma^\perp)$, we obtain:
\begin{equation}
	\begin{split}
		\norm{P(z)-x}{2}	& \leq L\norm{z-\orth(z)}{2}+\beta^{\perp}\norm{z-x}{2}\\
		& \leq L\norm{z-x}{2}+\beta^{\perp}\norm{z-x}{2} \text{~by definition of~} \orth \\ 			
		& \leq (L+\beta^{\perp})\norm{z-x}{2}\\
	\end{split}
\end{equation}
which concludes the proof.
\end{proof}

\section{Proof of Theorem \ref{theorem:eta_L}}

\begin{proof}[Proof of theorem \ref{theorem:eta_L}]

As $\Sigma$ is a homogeneous model and it can be decomposed as $ \Sigma=\underset{i\in I}{\bigcup} \spn(x_i)$, where $I$ is a finite or non finite set of indices and $x_i\in\bR^n$ such that $\norm{x_i}{2}=1$. 

Let $z\in \bR^n$. If $z\in\Sigma$, $\norm{P(z)-\orth(z)}{2}=\norm{z-\orth(z)}{2}=0$, any choice of $L$ verifies the hypothesis of Theorem~\ref{theorem:general_rlc}.

We suppose $z\in\bR^n \setminus \Sigma$ from now on. Let $i \in I$ and $\lambda \in \bR$ such that $P(z)=\lambda x_i$. Without loss of generality, we consider $\lambda>0$. We consider $\orthi$ the orthogonal projection over $V_i\myeq \spn(x_i)$. Therefore, $\orthi(z)=\langle z,x_i \rangle x_i$. We also consider $k\in I$ such that $\orth(z)=P_{V_k}^{\perp}(z)=\langle z,x_k \rangle x_k$ (since the orthogonal projection is supposed non-empty).
By the triangular inequality, we have:
\begin{equation}\label{eq:ineq_tri}
	\begin{split}
		\frac{\norm{P(z)-\orth(z)}{2}}{\norm{z-\orth(z)}{2}} & \leq  \frac{\norm{P(z)-\orthi(z)}{2}}{\norm{z-\orth(z)}{2}}+\frac{\norm{\orthi(z)-\orth(z)}{2}}{\norm{z-\orth(z)}{2}}\\
		& =  \frac{\norm{P(z)-\orthi(z)}{2}}{\norm{z-P(z)}{2}}\frac{\norm{z-P(z)}{2}}{\norm{z-\orth(z)}{2}}+\frac{\norm{\orthi(z)-\orth(z)}{2}}{\norm{z-\orth(z)}{2}}.\\
	\end{split}
\end{equation}

First: 
\begin{equation*}
	\begin{split}
		\frac{\norm{P(z)-\orthi(z)}{2}}{\norm{z-P(z)}{2}}&=\frac{\norm{\lambda x_i-\langle z,x_i\rangle x_i}{2}}{\norm{z-P(z)}{2}}=\frac{|\lambda-\langle z,x_i\rangle |}{\norm{z-P(z)}{2}}.\\
	\end{split}
\end{equation*}

We remark that  $\psi_P(z)=\frac{|\langle P(z),z-P(z)\rangle|}{\norm{P(z)}{2}\norm{z-P(z)}{2}}=\frac{|\langle \lambda x_i,z-\lambda x_i\rangle|}{\norm{\lambda x_i}{2}\norm{z-P(z)}{2}}=\frac{|\langle x_i,z-\lambda x_i\rangle|}{\norm{z-P(z)}{2}}=\frac{|\langle x_i,z\rangle -\lambda|}{\norm{z-P(z)}{2}}$, which gives:

\begin{equation}\label{eq:1}
	\begin{split}
		\frac{\norm{P(z)-\orthi(z)}{2}}{\norm{z-P(z)}{2}}&=\psi_P(z).\\
	\end{split}
\end{equation} 

Next, we focus on bounding $\frac{\norm{z-P(z)}{2}}{\norm{z-\orth(z)}{2}}$. We denote by $A(z):=1-\frac{\norm{z-\orth(z)}{2}^2}{\norm{z-P(z)}{2}^2}$ such that $\frac{\norm{z-P(z)}{2}^2}{\norm{z-\orth(z)}{2}^2}=\frac{1}{1-A(z)}$. Then:

\begin{equation}
	\begin{split}
		A(z)=1-\frac{\norm{z-\orth(z)}{2}^2}{\norm{z-P(z)}{2}^2}&=\frac{\norm{z-P(z)}{2}^2-\norm{z-\orth(z)}{2}^2}{\norm{z-P(z)}{2}^2}\\
		&=\frac{\norm{\lambda x_i-z}{2}^2-\norm{\langle x_k,z \rangle x_k-z}{2}^2}{\norm{z-P(z)}{2}^2}\\
		&=\frac{\lambda^2-2\lambda\langle x_i,z \rangle+\norm{z}{2}^2 -\langle x_k,z \rangle^2+2\langle x_k,z \rangle^2-\norm{z}{2} ^2}{\norm{z-P(z)}{2}^2}\\
		&=\frac{(\lambda-\langle z,x_i\rangle )^2}{\norm{z-P(z)}{2}^2}+  \frac{\langle x_k,z \rangle^2- \langle x_i,z \rangle^2}{\norm{z-P(z)}{2}^2}.\\
	\end{split}
\end{equation}
The first term is equal to $\psi_P(z)^2$, and the second term can be bounded as follows. With the Cauchy-Schwarz inequality,
\begin{equation*}
	\begin{split}
		\frac{\langle x_k,z \rangle^2- \langle x_i,z \rangle^2}{\norm{z-P(z)}{2}^2}&= \frac{(\langle x_k+ \langle x_i,z \rangle)(\langle x_k,z \rangle- \langle x_i,z \rangle)}{\norm{z-P(z)}{2}^2}\\
		&= \frac{(\langle x_k+ x_i,z \rangle)(\langle x_k- x_i,z \rangle)}{\norm{z-P(z)}{2}^2}\\
		&\leq \frac{\norm{z}{2}^2\norm{x_k+x_i}{2}\norm{x_k-x_i}{2}}{\norm{z-P(z)}{2}^2}\\
		&=\frac{\norm{z}{2}^2\sqrt{1+1-2\langle x_k,x_i\rangle}\sqrt{1+1+2\langle x_k,x_i\rangle}}{\norm{z-P(z)}{2}^2}\\
		&= \frac{2\norm{z}{2}^2\sqrt{1-\langle x_k,x_i\rangle}\sqrt{1+\langle x_k,x_i\rangle}}{\norm{z-\orth(z)}{2}^2}\\				
		&= \frac{2\norm{z}{2}^2\sqrt{1-\langle x_k,x_i\rangle^2}}{\norm{z}{2}^2-\langle z,x_k \rangle^2}=\frac{2\sqrt{1-\langle x_k,x_i\rangle^2}}{1-\langle \frac{z}{\norm{z}{2}},x_k \rangle^2}.\\
	\end{split}
\end{equation*} 
We can write $\langle x_k,x_i\rangle=\alpha(\orth(z),P(z))$ and $\langle x_k,\frac{z}{\norm{z}{2}}\rangle=\alpha(\orth(z),z)$. Thus:
\begin{equation}
	\frac{\langle x_k,z \rangle^2- \langle x_i,z \rangle^2}{\norm{z-P(z)}{2}^2}=\frac{2\sqrt{1-\alpha(\orth(z),P(z))^2}}{1-\alpha(\orth(z),z)^2}:=\phi^2_P(z).
\end{equation}
Therefore $A(z)\leq \psi_P(z)^2+\phi_P(z)^2$. 	As we supposed $ \underset{z\in\mathbb{R}^n\setminus \Sigma}{\sup} \psi_P(z)^2+\underset{z\in\mathbb{R}^n\setminus \Sigma}{\sup} \phi_P^2<1$, then $A(z)\leq \psi_P(z)^2+\phi_P(z)^2<1$ and we can write that:
\begin{equation}\label{eq:2}
	\frac{\norm{z-P(z)}{2}^2}{\norm{z-\orth(z)}{2}^2}=\frac{1}{1-A(z)}\leq \frac{1}{1-\psi_P(z)^2-\phi_P(z)^2}.
\end{equation}

Lastly, we compute the third term in \ref{eq:ineq_tri}. We have:
\begin{equation}
	\begin{split}
		\frac{\norm{\orthi(z)-\orth(z)}{2}^2}{\norm{z-\orth(z)}{2}^2}&=	\frac{\norm{\langle x_i,z \rangle x_i-\langle x_k,z \rangle x_k}{2}^2}{\norm{z-\orth(z)}{2}^2}\\
		&=\frac{\langle x_i,z\rangle^2 +\langle x_k,z\rangle^2-2\langle x_k,z\rangle\langle x_i,z\rangle\langle x_k,x_i \rangle}{\norm{z-\orth(z)}{2}^2}\\
		&\leq \frac{2\langle x_k,z\rangle^2-2\langle x_k,z\rangle\langle x_i,z\rangle\langle x_k,x_i \rangle }{\norm{z-\orth(z)}{2}^2}~,\text{because~} \norm{\orth(z)}{2}\geq \norm{\orthi}{2}\\
		&= \frac{2\langle x_k,z\rangle(\langle x_k,z\rangle-\langle x_i,z\rangle\langle x_k,x_i \rangle )}{\norm{z-\orth(z)}{2}^2}.\\
	\end{split}
\end{equation}
With the Cauchy-Schwarz inequality,					
\begin{equation}
	\begin{split}
		\frac{\norm{\orthi(z)-\orth(z)}{2}^2}{\norm{z-\orth(z)}{2}^2}&\leq \frac{2\norm{z}{2}|\langle z,x_k-\langle x_k,x_i \rangle x_i \rangle|}{\norm{z-\orth(z)}{2}^2}\\
		&\leq \frac{2\norm{z}{2}\cdot \norm{z}{2} \norm{x_k-\langle x_k,x_i \rangle x_i}{2} }{\norm{z-\orth(z)}{2}^2}\\
		&= \frac{2\norm{z}{2}^2\sqrt{\norm{x_k-\langle x_k,x_i \rangle x_i}{2}^2 }}{\norm{z-\orth(z)}{2}^2}\\
		&= \frac{2\norm{z}{2}^2\sqrt{1+\langle x_k,x_i \rangle^2-2\langle x_k,x_i \rangle^2} }{\norm{z-\orth(z)}{2}^2} \\
		&= \frac{2\norm{z}{2}^2\sqrt{1-\langle x_k,x_i \rangle^2}} {\norm{z-\orth(z)}{2}^2}=\frac{2\sqrt{1-\langle x_k,x_i\rangle^2}}{1-\langle \frac{z}{\norm{z}{2}},x_k \rangle^2}=\phi_P(z)^2.\\
	\end{split}
\end{equation}

This implies that:
\begin{equation}\label{eq:3}
	\frac{\norm{\orthi(z)-\orth(z)}{2}^2}{\norm{z-\orth(z)}{2}^2} \leq \\
	\phi(z)^2.
\end{equation}

By combining \ref{eq:1}, \ref{eq:2} and \ref{eq:3}, we conclude: 

\begin{equation*}
	\begin{split}
		\frac{\norm{P(z)-\orth(z)}{2}}{\norm{z-\orth(z)}{2}}&\leq \frac{\psi_P(z)}{\sqrt{1-\psi_P(z)^2+\phi_P(z)^2}} + \phi_P(z) \leq \frac{\Psi_P}{\sqrt{1-\Psi_P^2-\Phi_P^2}} + \Phi_P:=L.\\
	\end{split}
\end{equation*}

Where $\Psi= \underset{z\in\bR^n\setminus \Sigma}{\sup}\psi_P(z)$ and 
$\Phi_P=\underset{z\in\bR^n\setminus \Sigma}{\sup}\phi_P(z)$.
\end{proof}

\section{Training ressources and architectures of neural networks}\label{section:architecture_reseaux}
The neural networks have been built with Pytorch and trained on a cluster with one NVIDIA A100 Tensor Core GPU. This same GPU was used to solve inverse problems with DPPs. Tables \ref{tab:ae-archi} and \ref{tab:unetres_arch} give the architectures of the used neural network.
\newpage
\begin{table}[h]
\centering
\small
\caption{Autoencoder architecture for MNIST images}
\begin{tabular}{lll}
	\toprule
	\textbf{Stage} & \textbf{Layer Type} & \textbf{Details} \\
	\midrule
	Input  & Conv2d             & $1 \rightarrow 4$, kernel $3\times3$, stride 2, padding 1,                \\
	&&LeakyReLU\\
	
	Encoder 1 & Downsample Conv2d   & $4 \rightarrow 8$, kernel $3\times3$, stride 2,                       \\
	&&LeakyReLU\\
	Flatten & Flatten & - \\
	MLP Encoder 1& Linear    & $392 \rightarrow 210$, LeakyReLU                   \\
	MLP Encoder 2 & Linear    & $210 \rightarrow 200$, LeakyReLU                   \\
	Latent code & & size 200\\
	MLP Decoder 1 & Linear    & $200 \rightarrow 210$, LeakyReLU                   \\
	MLP Decoder 2 & Linear    & $210 \rightarrow 392$, LeakyReLU                   \\
	Resize & Resize to (8,7,7) & - \\
	
	Decoder 1 & Conv2d   & $8 \rightarrow 8$, kernel $3\times3$, stride 1                       \\
	Upsample&Upsample bilinear mode of factor 2 & LeakyReLU \\
	
	Decoder 2 & Conv2d   & $4 \rightarrow 4$, kernel $3\times3$, stride 1                      \\
	Upsample&Upsample bilinear mode of factor 2 & LeakyReLU \\
	
	Output & Conv2d   & $4 \rightarrow 1$, kernel $3\times3$, stride 1                      \\
	
	\bottomrule
\end{tabular}
\label{tab:ae-archi}
\end{table}

\begin{table}[h]
\centering
\small
\caption{DRUNET Architecture for CelebA (and MRI)}
\begin{tabular}{llll}
	\toprule
	\textbf{Stage} & \textbf{Layer Type} & \textbf{Details} & \textbf{Skip Connection} \\
	\midrule
	Input   & Conv2d             & $3 (1)\rightarrow 64$, kernel $3\times3$, stride 1, padding 1             & –                 \\
	Down1   & ResBlock $\times 2$ & $64 \rightarrow 64$, kernel $3\times3$, ELU                             & $\rightarrow$ Up1 \\
	& Downsample Conv2d   & $64 \rightarrow 64(128)$, kernel $2\times2$, stride 2                      &                   \\
	Down2   & ResBlock $\times 2$ & $64 (128)\rightarrow 64(128)$, kernel $3\times3$, ELU                          & $\rightarrow$ Up2 \\
	& Downsample Conv2d   & $64(128) \rightarrow 128$, kernel $2\times2$, stride 2                     &                   \\
	Down3   & ResBlock $\times 2$ & $128 \rightarrow 128$, kernel $3\times3$, ELU                          & $\rightarrow$ Up3 \\
	& Downsample Conv2d   & $128 \rightarrow 128$, kernel $2\times2$, stride 2                     &                   \\
	Body    & ResBlock $\times 2$ & $128 \rightarrow 128$, kernel $3\times3$, ELU                          & –                 \\
	Up3     & ConvTranspose2d     & $128 \rightarrow 128$, kernel $2\times2$, stride 2                     & $\leftarrow$ Down3\\
	& ResBlock $\times 2$ & $128 \rightarrow 128$, kernel $3\times3$, ELU                          &                   \\
	Up2     & ConvTranspose2d     & $128 \rightarrow 64(128)$, kernel $2\times2$, stride 2                     & $\leftarrow$ Down2\\
	& ResBlock $\times 2$ & $64(128) \rightarrow 64(128)$, kernel $3\times3$, ELU                          &                   \\
	Up1     & ConvTranspose2d     & $64(128)\rightarrow 64$, kernel $2\times2$, stride 2                      & $\leftarrow$ Down1\\
	& ResBlock $\times 2$ & $64 \rightarrow 64$, kernel $3\times3$, ELU                            &                   \\
	Output  & Conv2d              & $64 \rightarrow 3(1)$, kernel $3\times3$, stride 1, padding 1             & –                 \\
	\bottomrule
\end{tabular}
\label{tab:unetres_arch}
\end{table}

\section{Test of SOR for an autoencoder over different images of $\mathbb{R}^n$ }\label{section:results_lambda}

In this section, we apply our SOR loss, as detailed in Section~\ref{section:section_ortho}, to an autoencoder trained on MNIST.

To test the efficiency of the SOR, we compare the mean $\psi_P$ defined in (\ref{eq:psi}) for each autoencoder at different distances from the model $\Sigma$: within, near, and far from $\Sigma$. To do this, we compute the average PSNR of MNIST reconstruction and the mean of $\psi_P$ for four different datasets: MNIST, MNIST with an additive Gaussian noise, Fashion MNIST dataset (28x28 clothes images) and images in $\mathcal{U}([0,1]^n)$. We train different autoencoders with several weights $\lambda$. Figure \ref{fig:results_lambda} in appendix \ref{section:results_lambda} represents the PSNR for MNIST reconstruction by the autoencoders and $\psi_P(z)$ as functions of the regularization coefficient $\lambda$. It shows that while the PSNRs remain acceptable with the increase of $\lambda$, $\psi_P$ decreases in each of the four datasets. This is particularly significant for datasets on which the autoencoder has not been trained (eg, Fashion MNIST). These autoencoders approximate better orthogonal projections. 

We choose $\lambda=0.4$ as an optimal tradeoff between the MSE and the SOR.

\begin{figure}[h]
\centering
\includegraphics[width=0.9\textwidth]{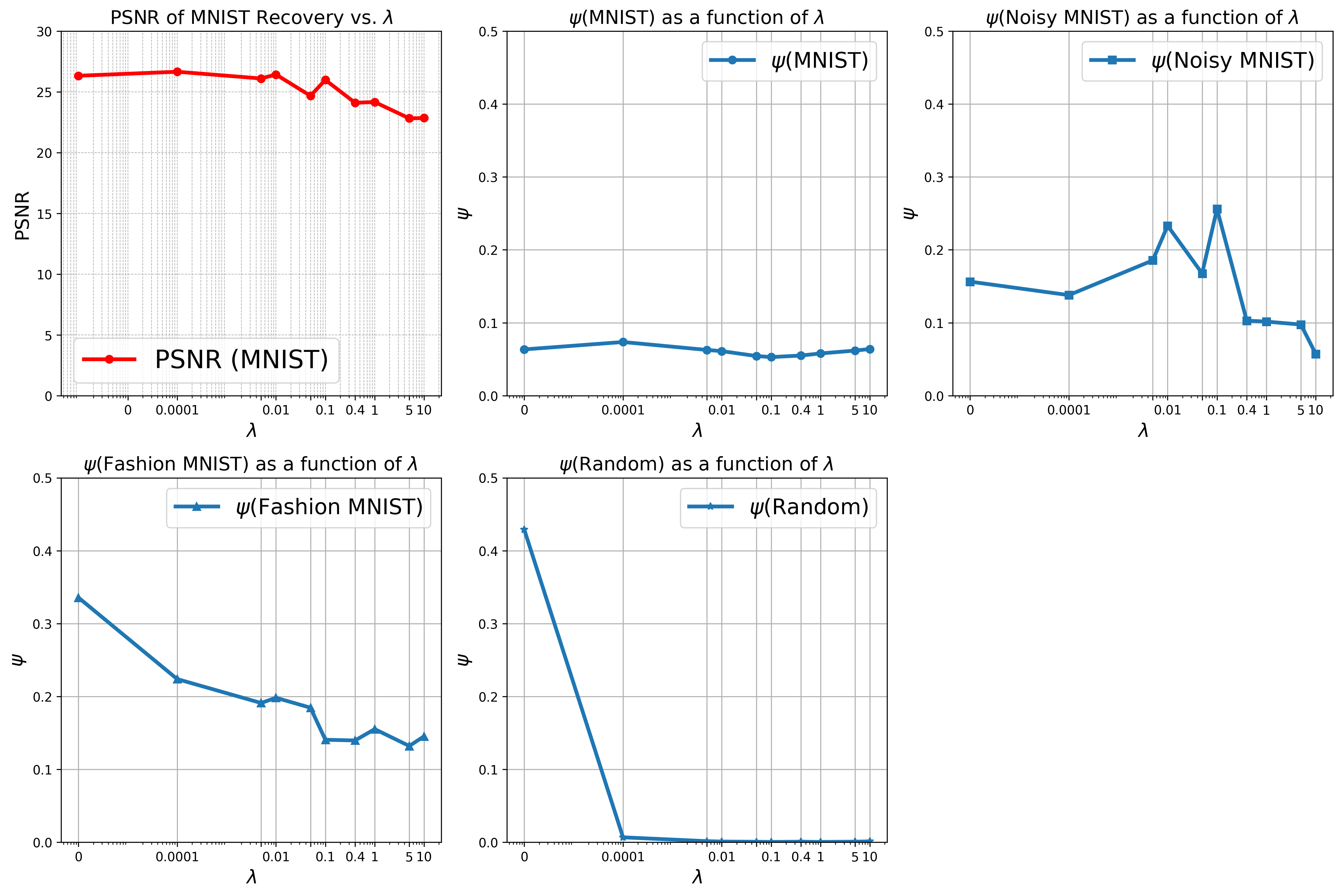}
\caption{PSNR of MNIST recovery and $\psi_P$ for four different datasets. $\psi_P$ is decreasing overall for the different datasets with respect to $\lambda$ while the PSNR is only slightly  degraded.}\label{fig:results_lambda}
\label{fig:orth_values}
\end{figure}

\section{SOR of Plug and Play DPP on MRI images}\label{section:pnp-mri}

We train a DRUNET denoiser on a dataset of size 9000 of MRI images \cite{ds003592:1.0.13}. The images represent brains scanned from different angles. They are grayscale images resized to 256x256. We considered for these, two denoisers without SOR ($\lambda=0$) and with SOR ($\lambda=0.1$, which we selected manually to give the best results) trained during 500 epochs.

\begin{figure}[h]
\centering
\begin{subfigure}{0.45\textwidth}
	\includegraphics[width=1\textwidth]{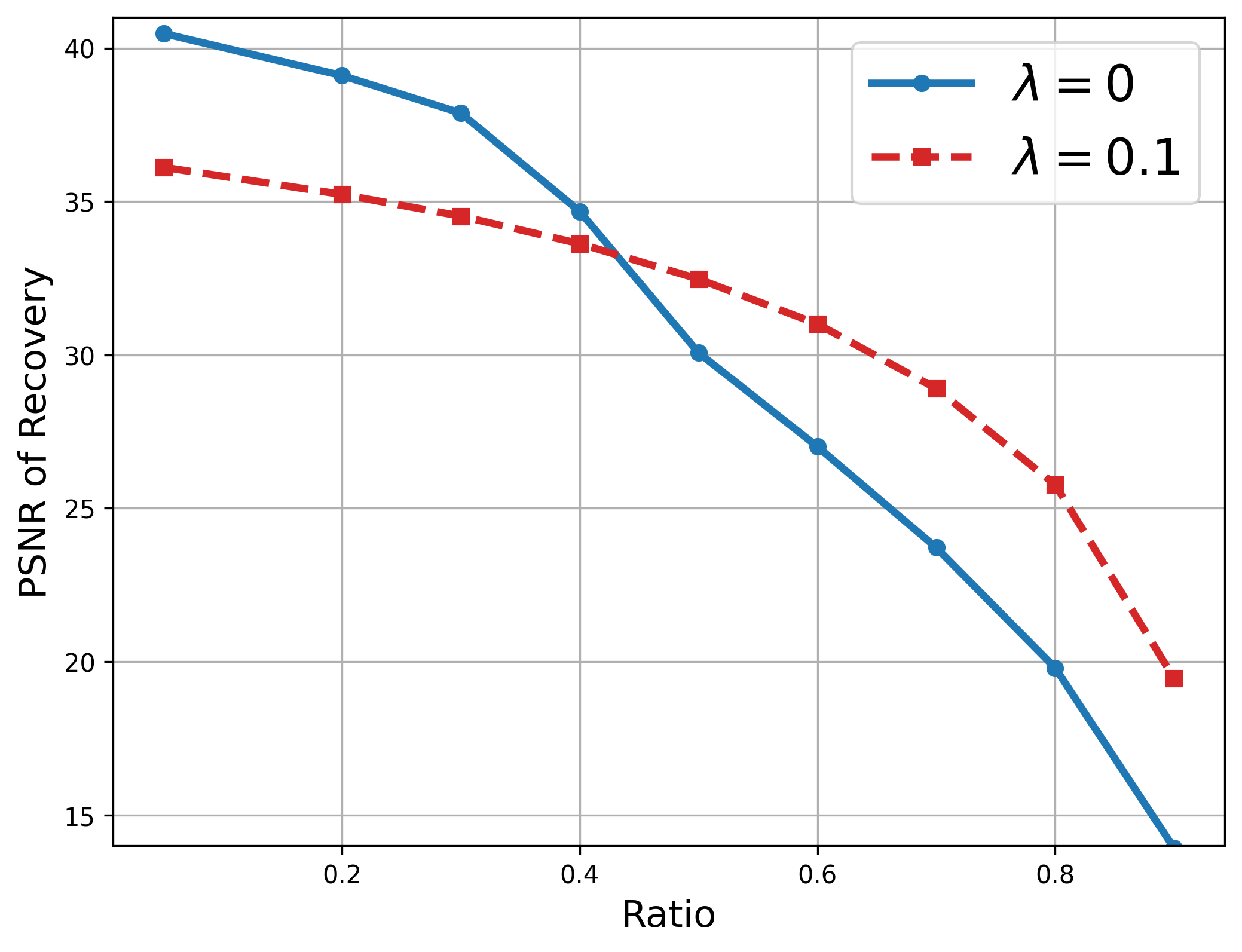}
	\caption*{Evolution of the mean PSNR of recovery as a function of the ratio of missing pixels}
\end{subfigure}
\hspace{0.1cm}
\begin{subfigure}{0.187\textwidth}
	\includegraphics[width=1\textwidth]{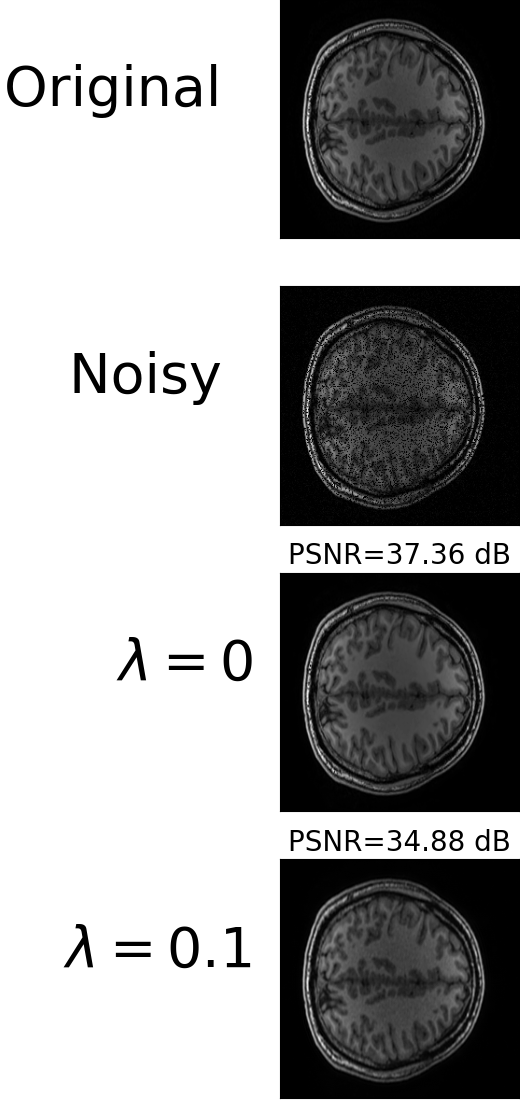}
	\caption*{~~~~~~Ratio: ~~~0.2}
\end{subfigure}
\hspace{0.1cm}
\begin{subfigure}{0.086\textwidth}
	\includegraphics[width=1\textwidth]{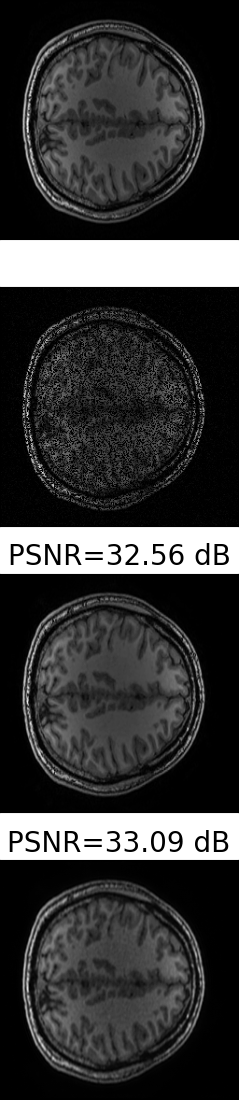}
	\caption*{0.4}
\end{subfigure}
\hspace{0.1cm}
\begin{subfigure}{0.0855\textwidth}
	\includegraphics[width=1\textwidth]{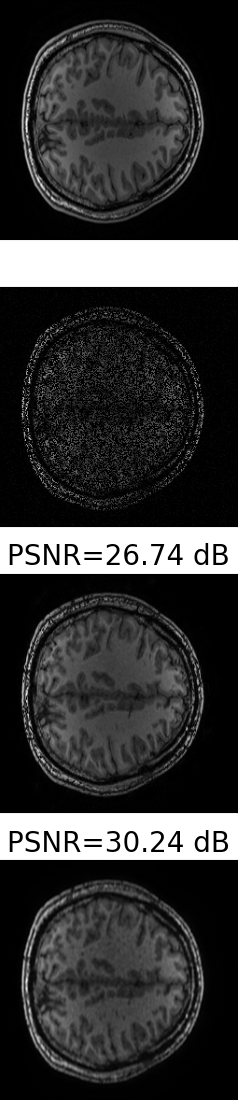}
	\caption*{0.6}
\end{subfigure}
\hspace{0.1cm}
\begin{subfigure}{0.086\textwidth}
	\includegraphics[width=1\textwidth]{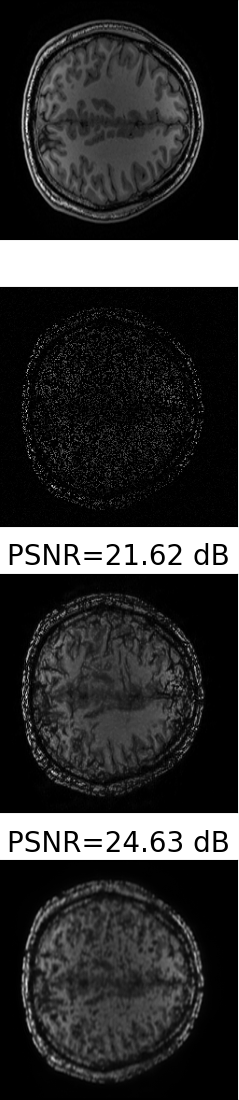}
	\caption*{0.8}
\end{subfigure}

\caption{Graph and visual results of the evolution of the recovered MRI images with the increase of the inpainting ratio for the different denoisers. In higher ratios, the regularized denoiser performs better than the non-regularized denoiser.}\label{fig:mri_graph_mask}
\end{figure}

Figure \ref{fig:mri_graph_mask} represents the recovery from a random inpainting problem and the evolution of the mean PSNR of recovered MRI images as a function of the inpainting ratio. Again, the mean recovery PSNR deteriorates far more quickly for the non-regularized denoiser ($\lambda=0$) compared to the regularized one ($\lambda=0.1$). At the highest ratios, the denoising prior trained with the SOR loss with $\lambda=0.1$ clearly outperforms the denoising prior with no SOR regularization, with up to a 3dB improvement for $80\%$ missing pixels. This once again highlights the robustness of SOR in challenging situations. 

Furthermore, Figures \ref{fig:sup_res_mri_2} and \ref{fig:sup_res_mri_3} represent the visual recoveries from a super-resolution inverse problem with two different factors ($2$ and $3$) for several images. When the super resolution factor is 2, the mean average PSNR recovery over the test set of the standard denoiser ($\lambda=0$) is 31.84 dB, with a convergence of 21.68 iterations. On the other hand, the regularized denoiser ($\lambda=0.1$) increases the average PSNR to 33.02 dB, and the convergence is faster, reducing to 16.58 iterations on average. In summary, SOR allows a faster convergence while providing acceptable (even better) recovery images.

With a sub-sampling by a factor of 3, keeping the Gaussian blur $\mathbf{F}$ unchanged (i.e. some aliasing is present in the subsampled images), the denoiser without SOR gives an average PSNR of 24.65 dB with a convergence in 51.47 iterations, whereas the regularized denoiser gives a PSNR of 28.52 dB with a faster convergence of 43.5 iterations. The regularized denoiser suffers less from the new augmented factor with a loss of 4.49 dB (from 33.02dB to 28.52dB) compared to 7.19 dB of the basic denoiser (31.84dB to 24.65dB). The particularity of this example is that the recovered images with $\lambda=0.1$ create less ringing effect compared to the $\lambda=0$, as can be seen in Figure \ref{fig:sup_res_mri_3}. This shows that the SOR adds clearly a robustness to aliasing which can be modeled as a structured noise (again through the stability constant given by Theorem~\ref{th:gen_convevergence}). In summary, we find that not only SOR loss speeds up the convergence of GPGD, but also allows a better identifiability of $\Sigma$, \emph{without $\Sigma$ ever being explicitly known}.

\begin{figure}[h]
\centering
\begin{subfigure}{0.4\textwidth}
	\includegraphics[width=\textwidth]{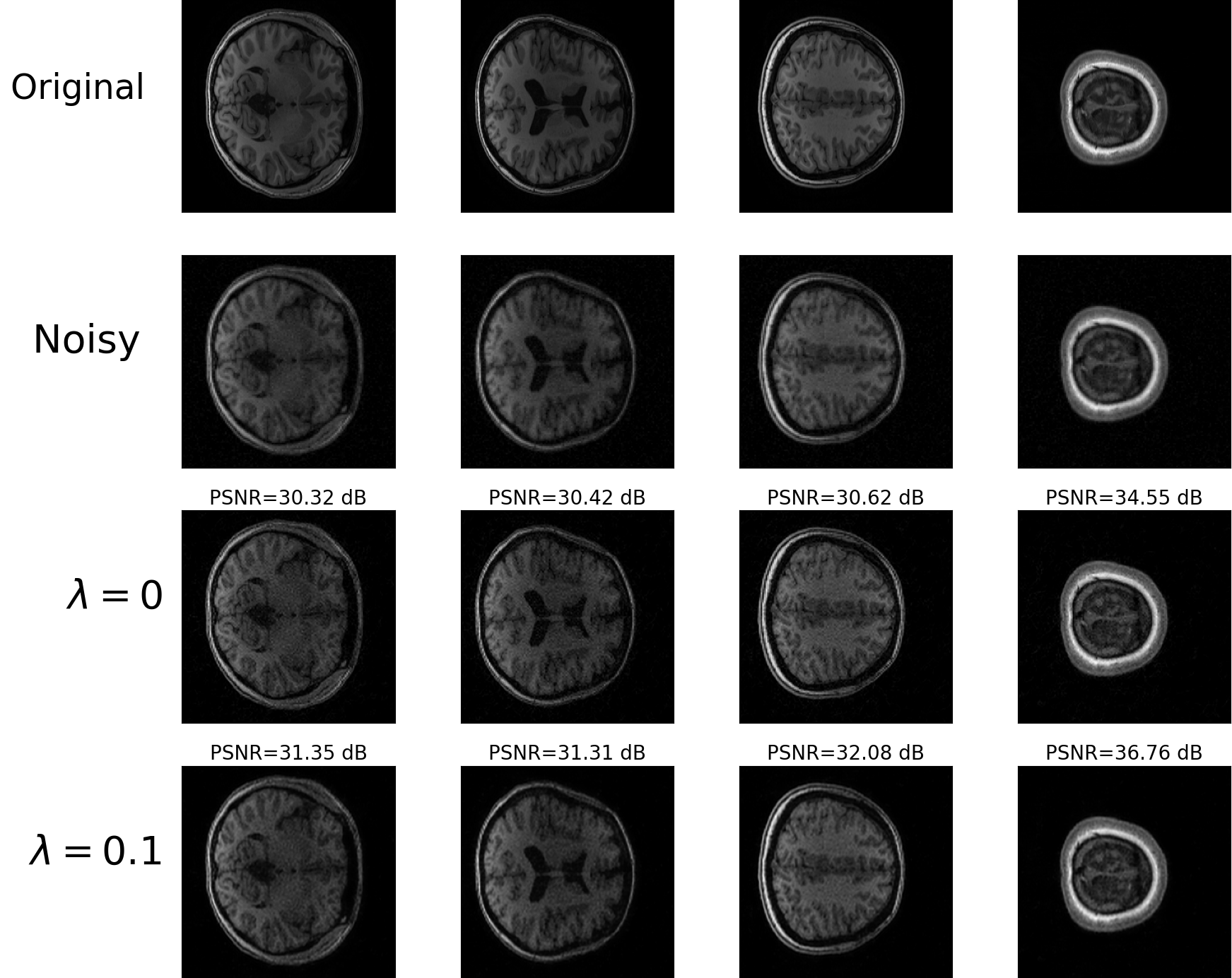}
	\caption{Super resolution of a factor 2. Mean PSNR: \textbf{31.84dB} for $\lambda=0$ and \textbf{33.02dB} for $\lambda=0.1$. $\lambda=0.1$ (faster by 25$\%$)}
	\label{fig:sup_res_mri_2}
\end{subfigure}
\hspace{0.1cm}
\begin{subfigure}{0.4\textwidth}
	\includegraphics[width=\textwidth]{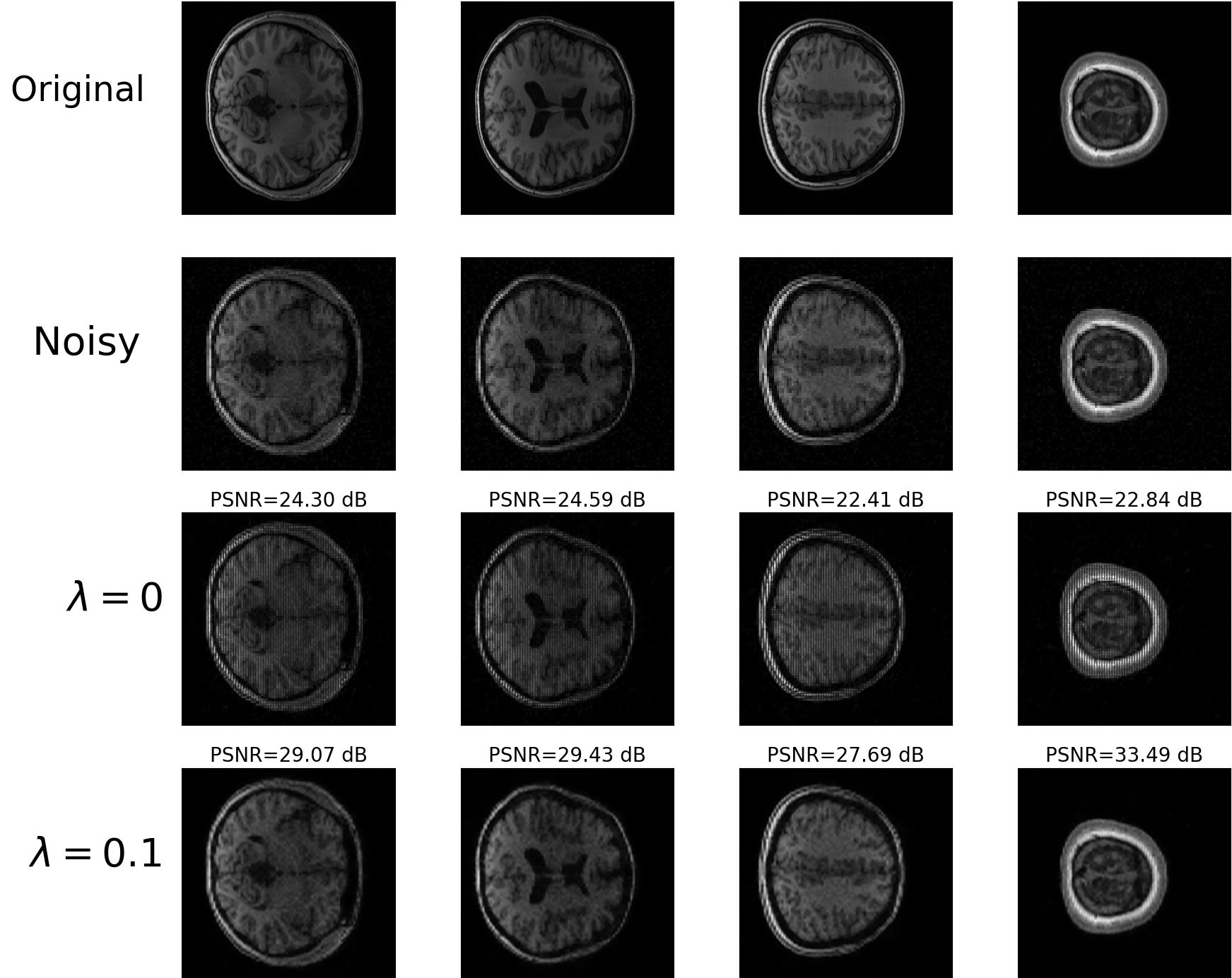}
	\caption{Super resolution of a factor 3. The mean PSNR is \textbf{24.65} for $\lambda=0$ and \textbf{28.52} for $\lambda=0.1$. $\lambda=0.1$ is quicker by 16$\%$}
	\label{fig:sup_res_mri_3}
\end{subfigure}
\caption{Recovery of MRI images for a super-resolution inverse problem. We can see that regularizing gives better recovery in both super-resolution cases and is more robust to aliasing.}
\end{figure}

\newpage 
\section{Additional experiments on MNIST}
\subsection{Inverse problems for a noiseless model ($\sigma=0$)}

We consider in this section inverse problems where the model (\ref{eq:eq_pb_inverse}) is a noiseless model, \text{i.e.} $\sigma=0$. We carry out the same experiments as in the previous noisy model ($\sigma=0.2$).

Figures \ref{fig:convergence_speed_noiseless} and \ref{fig:Mask_random_MNIST_noiseless}  represent respectively the recoveries for a super-resolution inverse problem and for a missing pixel inpainting inverse problem. We consider different $\lambda$ of the regularized autoencoders.

\begin{figure}[!h]
\centering
\begin{subfigure}{0.45\textwidth}
	\includegraphics[width=1\textwidth]{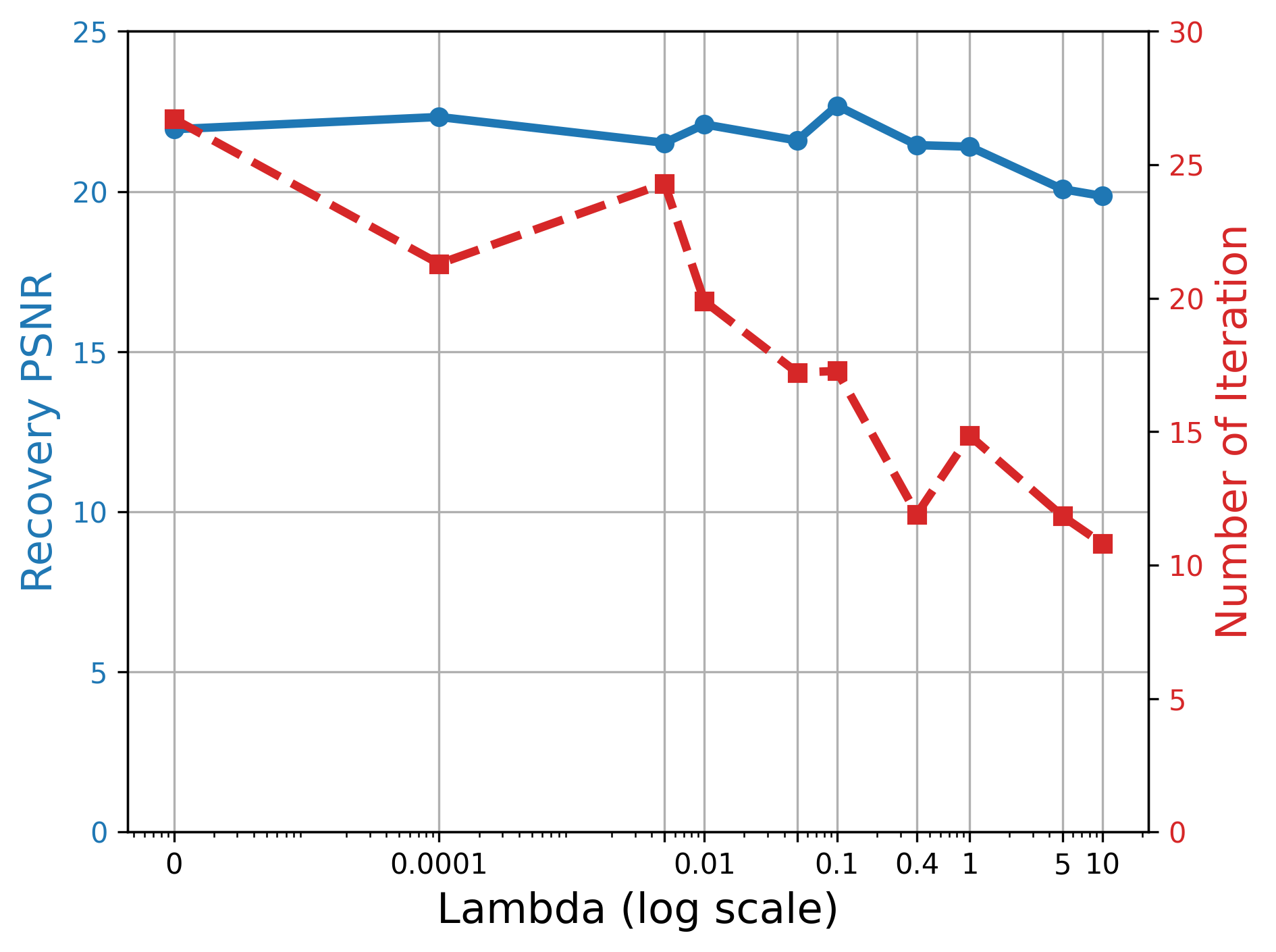}
	\caption*{Recovery PSNR and convergence speed as a function of the weight $\lambda$}\label{graph_supresmnist_noiseless}
\end{subfigure}
\hspace{0.1cm}
\begin{subfigure}{0.45\textwidth}
	\includegraphics[width=1\textwidth]{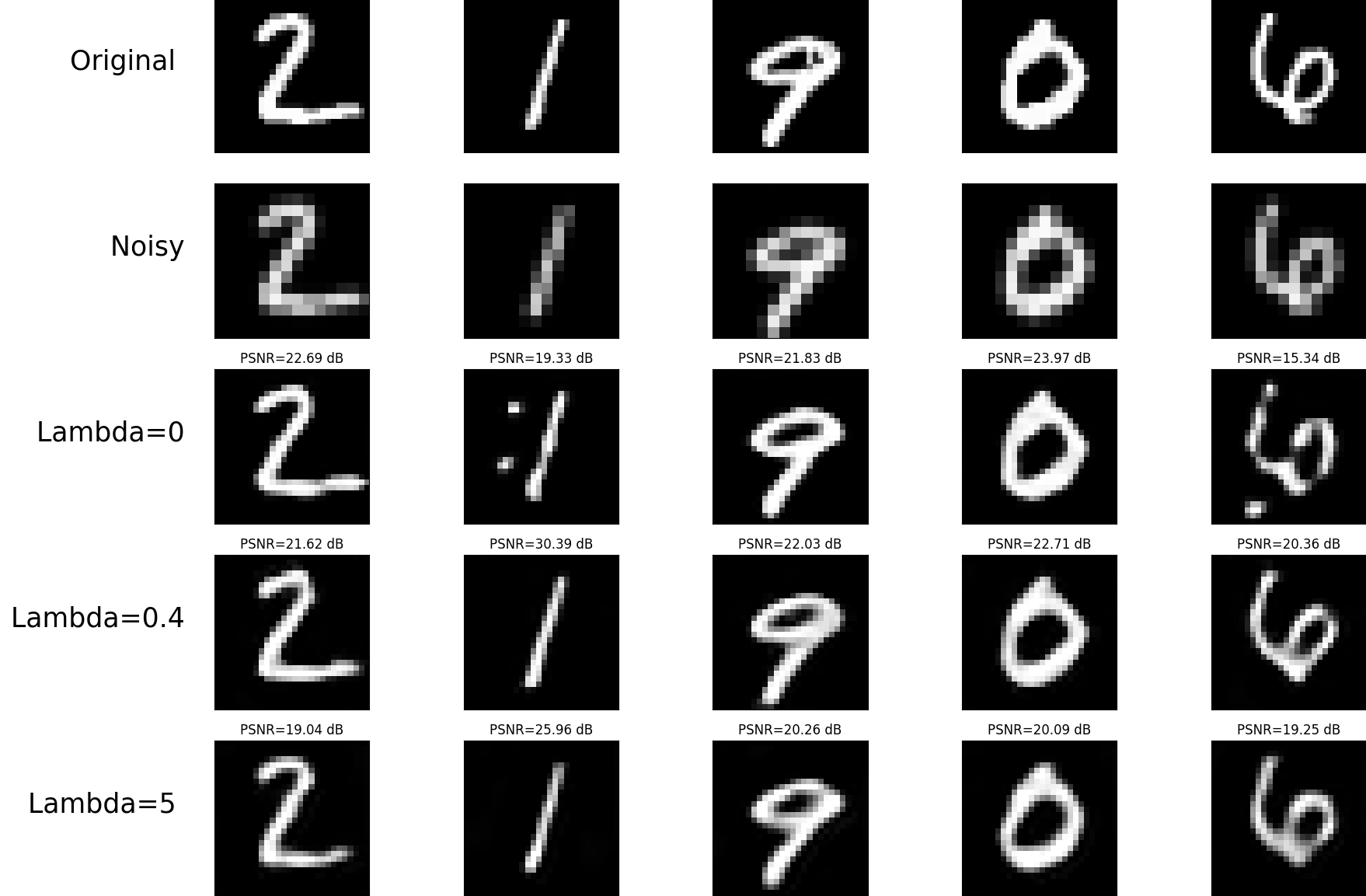}
	\caption*{Recovery of subsampled MNIST images through different PGDs.}
\end{subfigure}
\caption{Recovery of MNIST images from a noiseless super resolution inverse problem with different PGDs using autoencoders weighted differently }\label{fig:convergence_speed_noiseless}
\end{figure}

As previously, Figure \ref{fig:convergence_speed_noiseless} shows that the convergence of the PGD is clearly faster when increasing $\lambda$ while the recovery PSNR remains stable, which indicates that SOR allows an acceleration of PGDs without a significant loss of its performance.

\begin{figure}[!h]
\centering
\begin{subfigure}{0.4\textwidth}
	\includegraphics[width=1\textwidth]{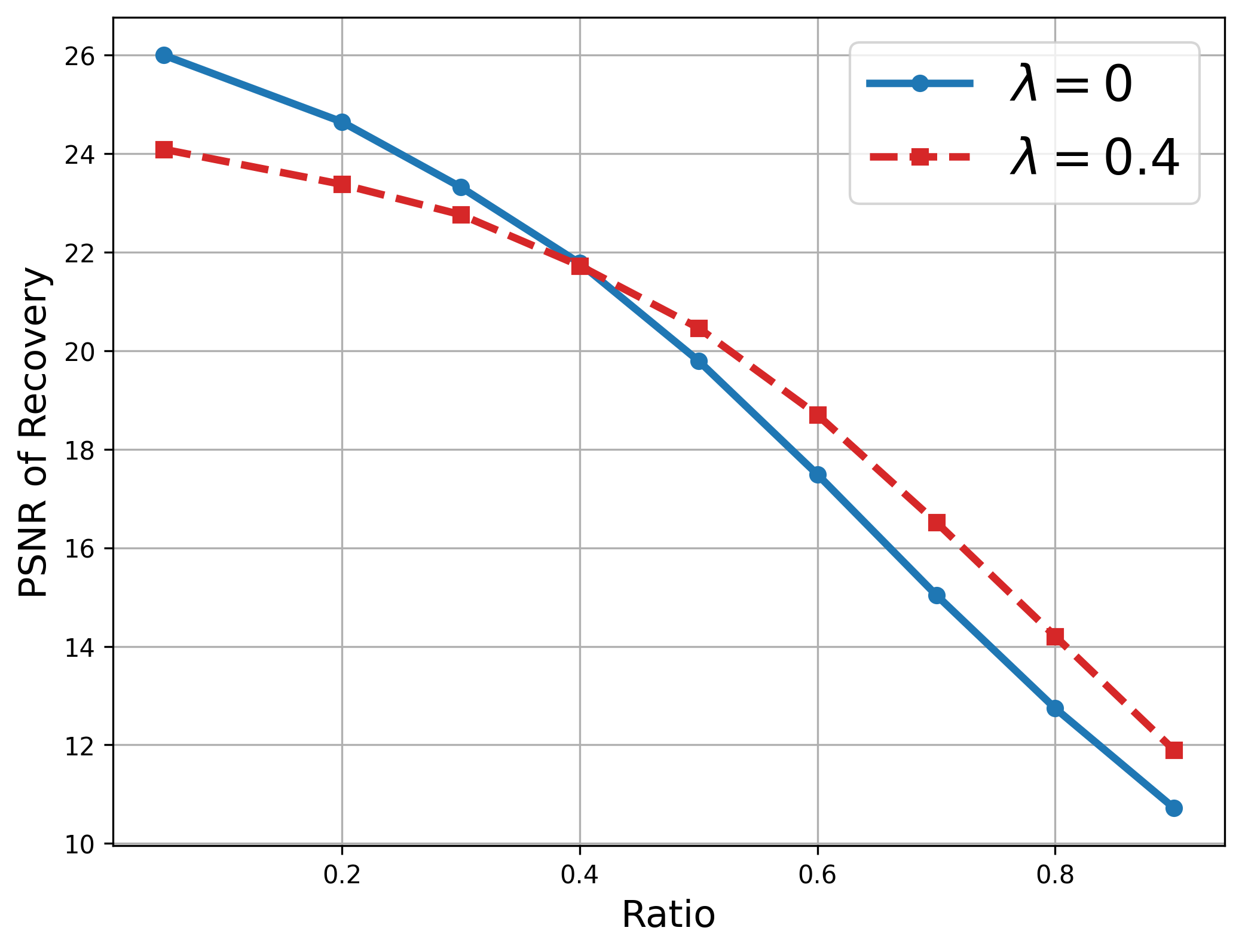}
	\caption*{Evolution of the mean PSNR of recovery as a function of the ratio}\label{fig:graph_MNIST_mask_random_noiseless}
\end{subfigure}
\hspace{0.33cm}
\begin{subfigure}{0.1565\textwidth}
	\includegraphics[width=1\textwidth]{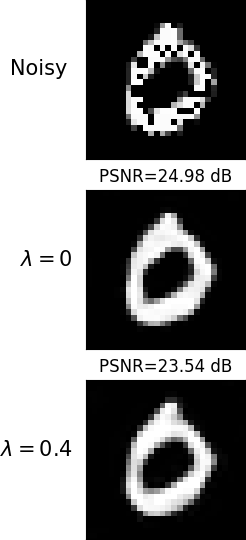}
	\caption*{~~~~~~~~~~Ratio = 0.2} 
\end{subfigure}
\hspace{0.1cm}
\begin{subfigure}{0.1\textwidth}
	\includegraphics[width=1\textwidth]{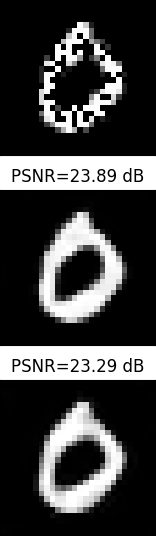}
	\caption*{Ratio = 0.4}
\end{subfigure}
\hspace{0.1cm}
\begin{subfigure}{0.1\textwidth}
	\includegraphics[width=1\textwidth]{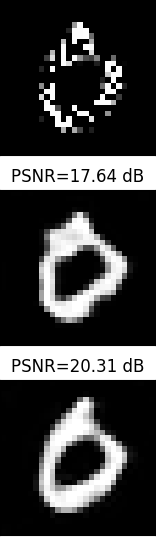}
	\caption*{Ratio = 0.6}
\end{subfigure}
\hspace{0.1cm}
\begin{subfigure}{0.1\textwidth}
	\includegraphics[width=1\textwidth]{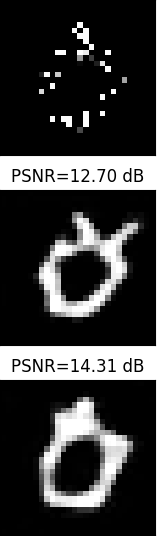}
	\caption*{Ratio = 0.8}
\end{subfigure}

\caption{Graph and visual results of the evolution of the recovered MNIST image from a noiseless inpainting problem with the increase of the inpainting ratio for the different autoencoders. }\label{fig:Mask_random_MNIST_noiseless}
\end{figure}

For the missing pixel inpainting, Figure \ref{fig:Mask_random_MNIST_noiseless}, the autoencoder with $\lambda=0.4$ is clearly more stable to the deterioration of the measurement operator. This experiment proves the robustness of the PGD to the deterioration of $\mathbf{A}$ when using a regularized DPP.

\subsection{Noisy inverse problems with  $\sigma=0.05$}

We consider in this section inverse problems where the noise has a greater $\sigma$: $\sigma=0.05$ instead of $\sigma=0.02$. We carry out the same experiments as in the previous cases.

Figures \ref{fig:convergence_speed_noisy} and \ref{fig:Mask_random_MNIST_noisy}  represent respectively the recoveries from a super-resolution inverse problem and from a missing pixel inpainting inverse problem. We consider different $\lambda$ for the regularized autoencoders. The same observations as in the cases $\sigma=0$ and $\sigma=0.02$ are still valid.

\begin{figure}[!h]
\centering
\begin{subfigure}{0.45\textwidth}
	\includegraphics[width=1\textwidth]{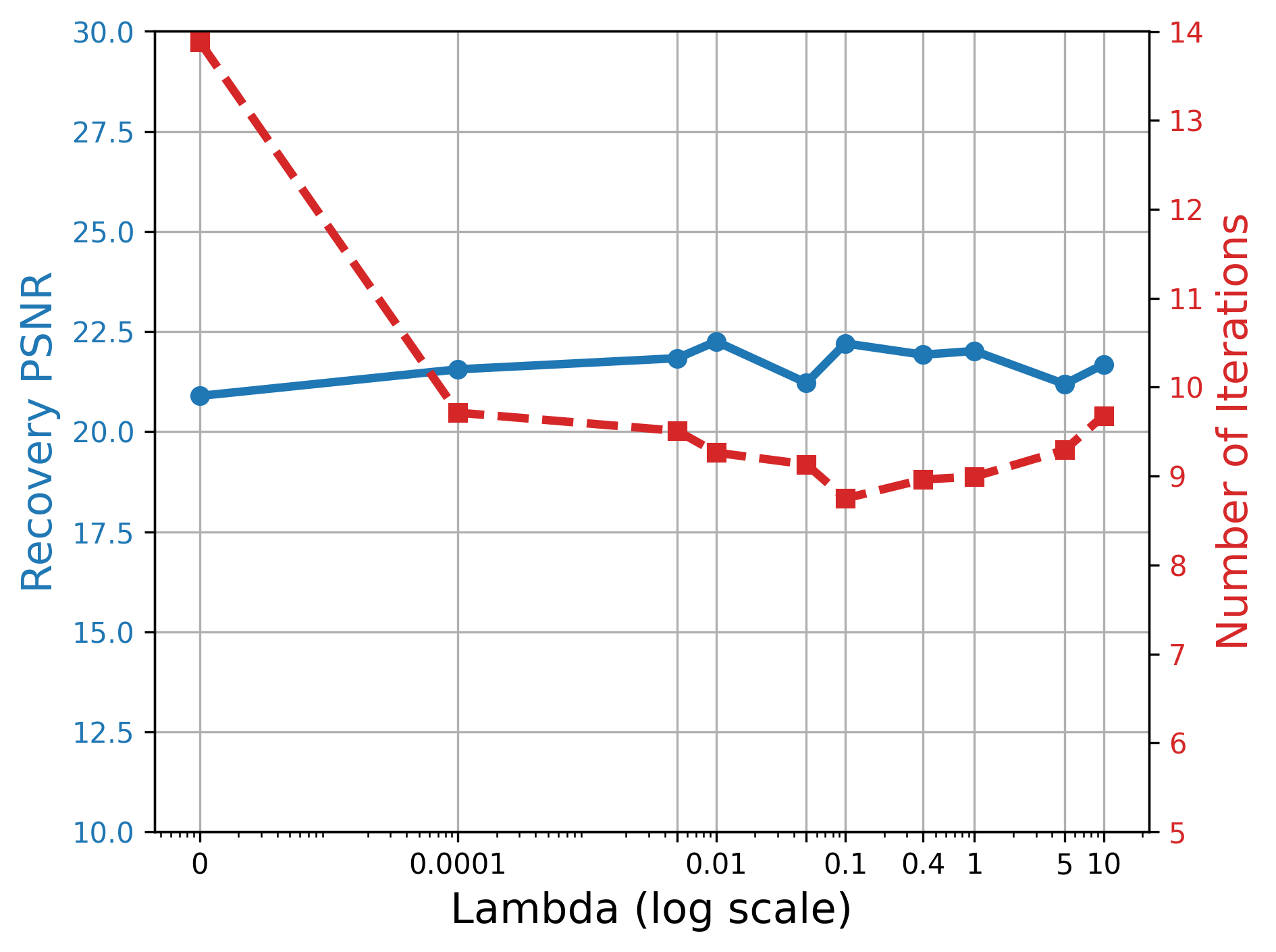}
	\caption*{Recovery PSNR and convergence speed as a function of the weight $\lambda$}\label{graph_supresmnist_noisy}
\end{subfigure}
\hspace{0.1cm}
\begin{subfigure}{0.45\textwidth}
	\includegraphics[width=1\textwidth]{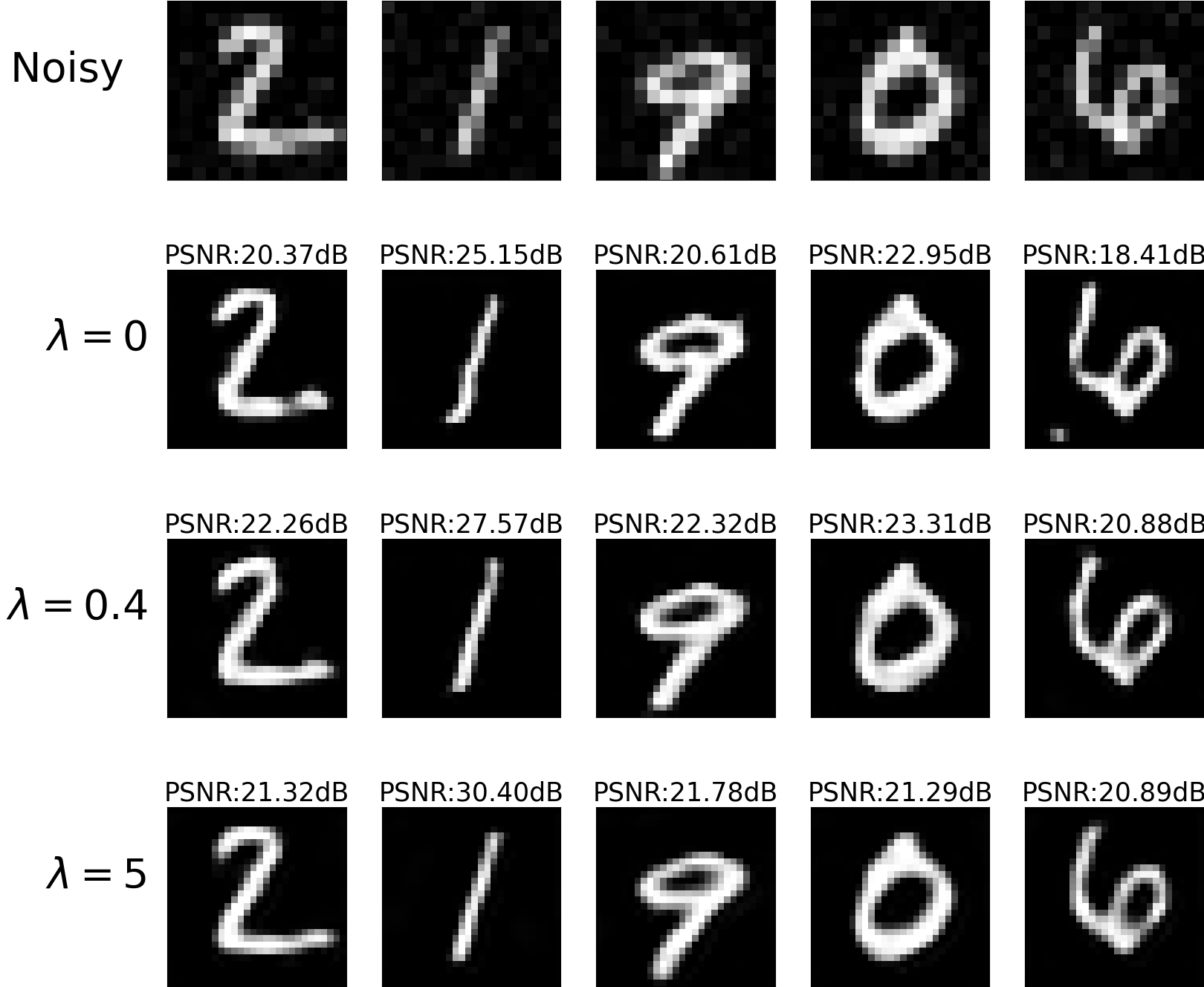}
	\caption*{Recovery of subsampled MNIST images through different PGDs.}
\end{subfigure}
\caption{Recovery of MNIST images from a noisy super-resolution inverse problem with different PGDs using autoencoders weighted differently. The level of the noise is 0.05. }\label{fig:convergence_speed_noisy}
\end{figure}

\begin{figure}[!h]
\centering
\begin{subfigure}{0.4\textwidth}
	\includegraphics[width=1\textwidth]{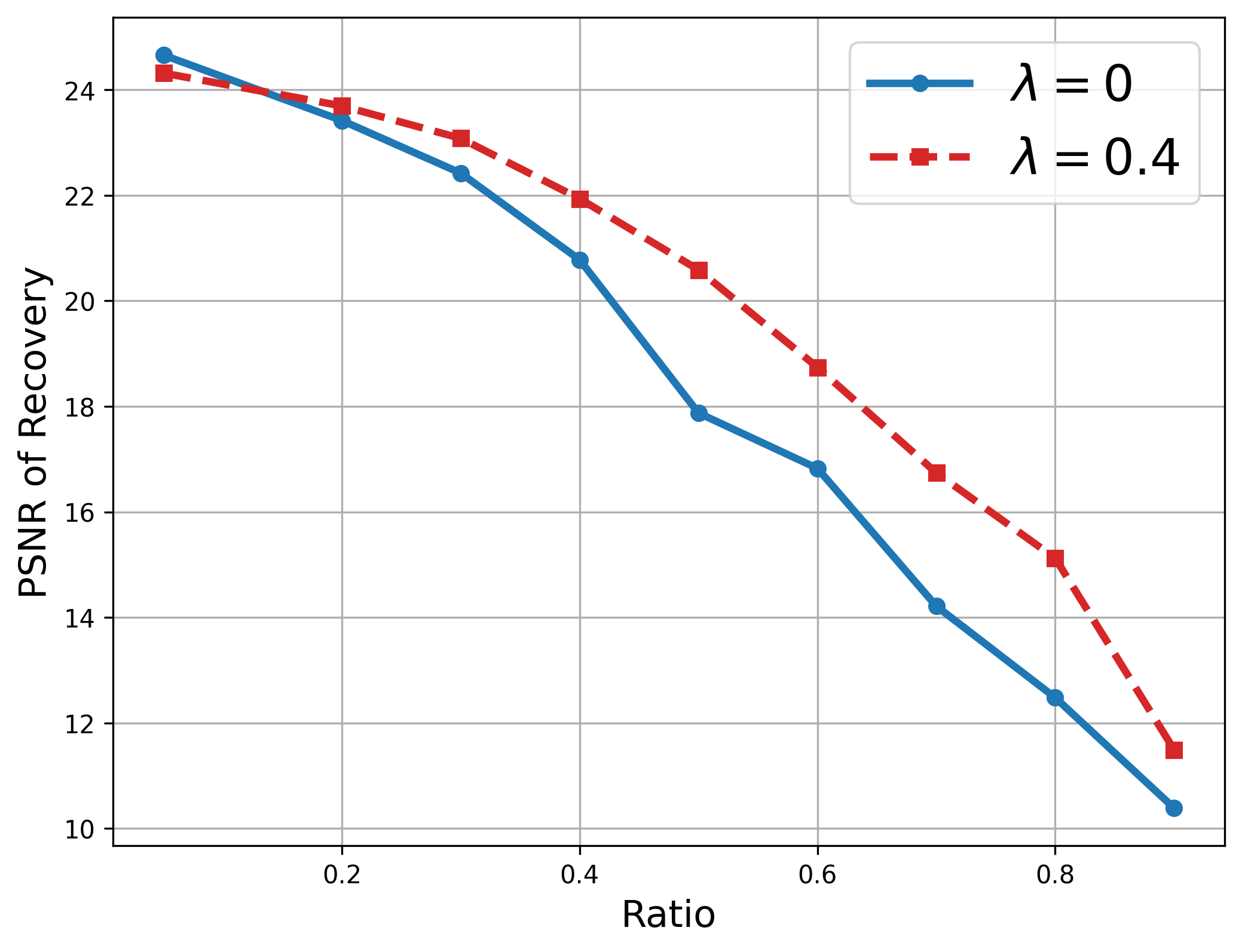}
	\caption*{Evolution of the mean PSNR of recovery as a function of the ratio}\label{fig:graph_MNIST_mask_random_noisy}
\end{subfigure}
\hspace{0.33cm}
\begin{subfigure}{0.1565\textwidth}
	\includegraphics[width=1\textwidth]{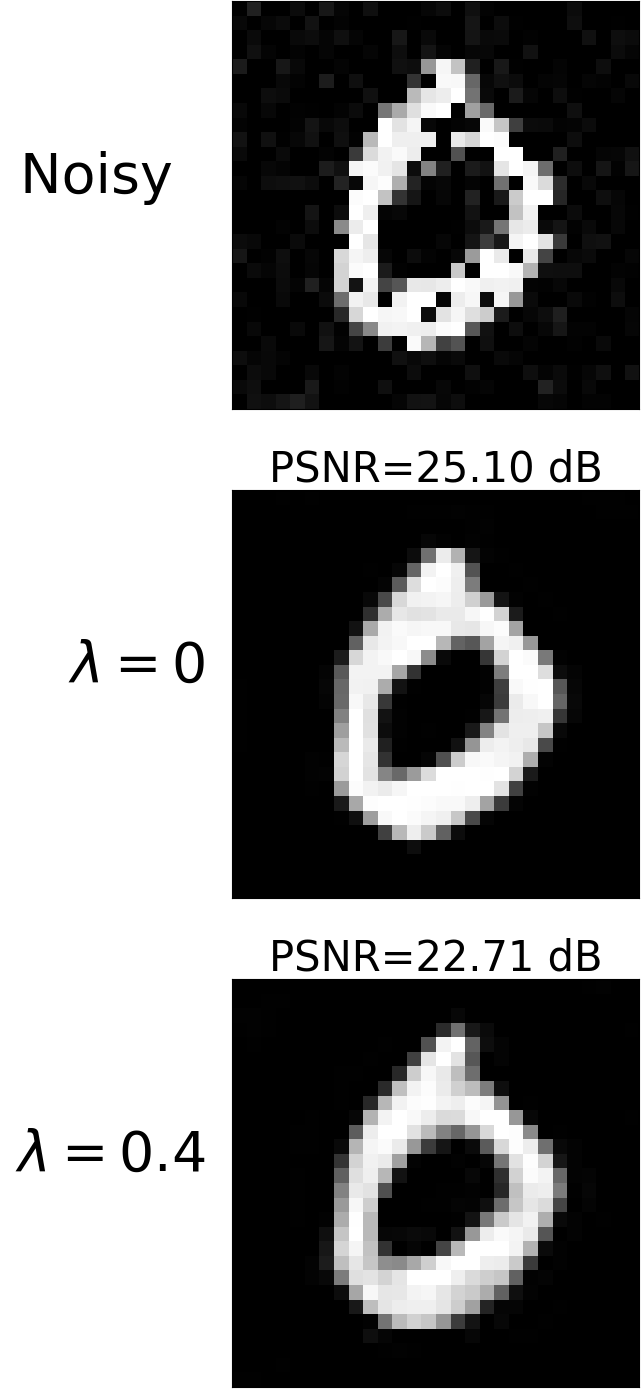}
	\caption*{~~~~~~~~~~Ratio = 0.2} 
\end{subfigure}
\hspace{0.1cm}
\begin{subfigure}{0.1\textwidth}
	\includegraphics[width=1\textwidth]{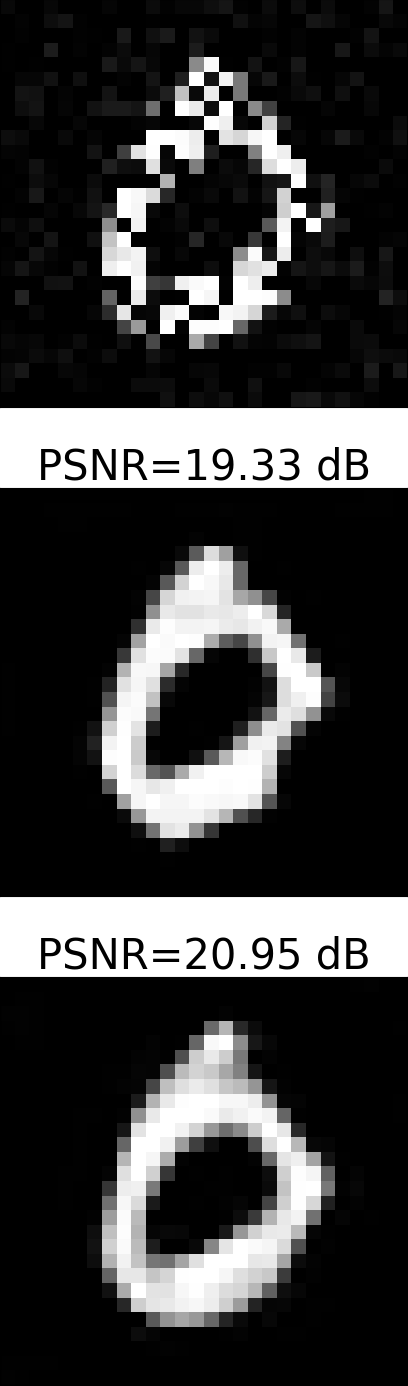}
	\caption*{Ratio = 0.4}
\end{subfigure}
\hspace{0.1cm}
\begin{subfigure}{0.1\textwidth}
	\includegraphics[width=1\textwidth]{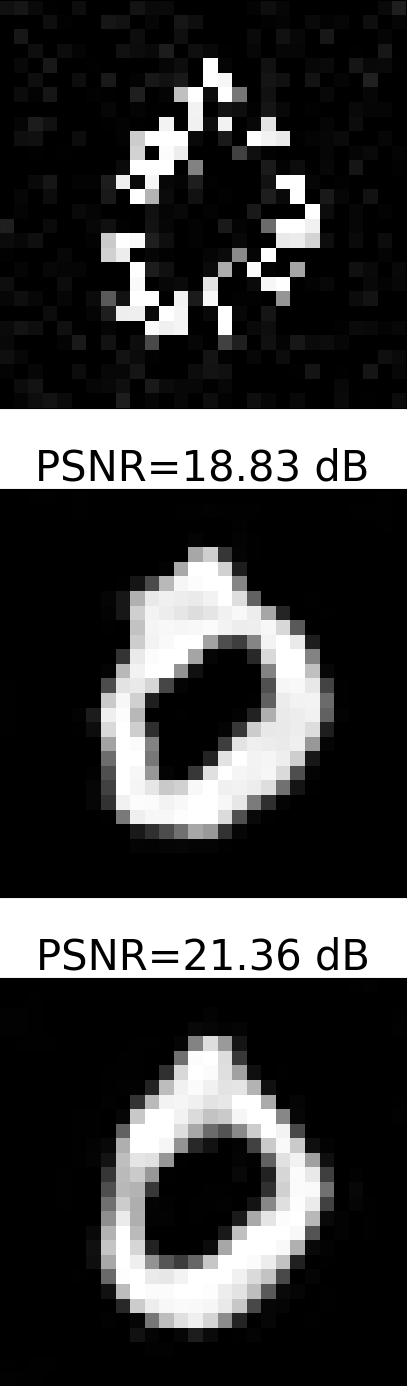}
	\caption*{Ratio = 0.6}
\end{subfigure}
\hspace{0.1cm}
\begin{subfigure}{0.1\textwidth}
	\includegraphics[width=1\textwidth]{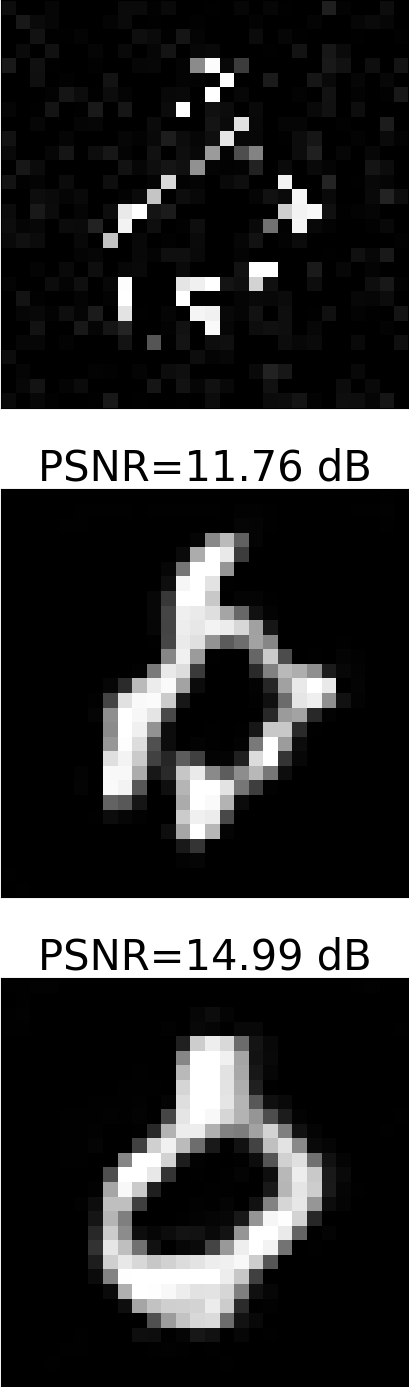}
	\caption*{Ratio = 0.8}
\end{subfigure}

\caption{Graph and visual results of the evolution of the recovered MNIST image with the increase of the inpainting ratio for the different autoencoders. The level of the noise is 0.05. }\label{fig:Mask_random_MNIST_noisy}
\end{figure}

\newpage 
\section{Additional experiments on CelebA}
In this section, we present additional experiments for CelebA. We increase, in particular, the noise level in each of them. We also added an inpainting inverse problem of ratio 0.7.
\subsection{Inpainting}
Table \ref{table:Inpainting_0.7} represents the metrics related to the inpainting inverse problem of ratio 0.7. As observed in \ref{table:Inpainting_0.4} and \ref{table:Inpainting_0.6}, SOR performs than PGD and RED in more challenging situations, \textbf{i.e.} when the noise and the inpainting ratio increase. In addition to that, it still converges faster compared to RED (as PGD struggles to recover acceptable images, see \ref{fig:inpainting_0.7_0} and \ref{fig:inpainting_0.7_0.02}).

\begin{table}[!h]
\centering
\small
\begin{tabular}{ccccc}
	\toprule
	$\sigma$           & Method & PSNR$\uparrow$  & SSIM$\uparrow$ & Iterations$\downarrow$\\
	\midrule
	\multirow{3}{*}{0}    & PGD    & 17.23 & 0.46 & \textbf{2.86}          \\
	& SOR    & \textbf{26.60} & \textbf{0.81} & 16.36         \\
	& RED    & 24.17 & 0.59 & 36.8         \\
	\hline
	\multirow{3}{*}{0.02} & PGD    & 17,32 & 0,46 & \textbf{2.88}         \\
	& SOR    & \textbf{26.61} & \textbf{0.81} & 19.2          \\
	& RED    & 24.06 & 0.58 & 37.76       
\end{tabular}
\caption{Metrics of an inpainting ratio of 0.7: SOR gives the best recovery results while still being quicker in both noiseless and noisy cases (compared to RED).}\label{table:Inpainting_0.7}
\end{table}

\begin{figure}[!h]
\centering
\begin{subfigure}{0.23\textwidth}
	\includegraphics[width=1\textwidth]{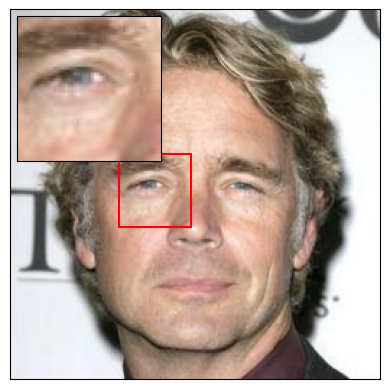}
	\caption{Original}
\end{subfigure}
\hspace{0.02cm}
\begin{subfigure}{0.23\textwidth}
	\includegraphics[width=1\textwidth]{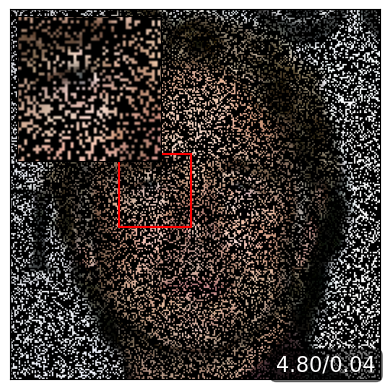}
	\caption{Observed}
\end{subfigure}
\hspace{0.02cm}
\begin{subfigure}{0.23\textwidth}
	\includegraphics[width=1\textwidth]{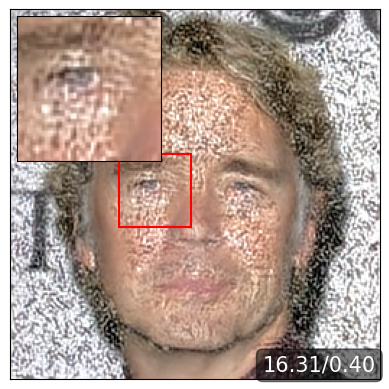}
	\caption{PGD}
\end{subfigure}

\begin{subfigure}{0.23\textwidth}
	\includegraphics[width=1\textwidth]{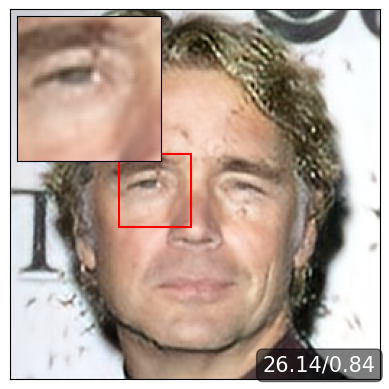}
	\caption{SOR}
\end{subfigure}
\hspace{0.02cm}
\begin{subfigure}{0.23\textwidth}
	\includegraphics[width=1\textwidth]{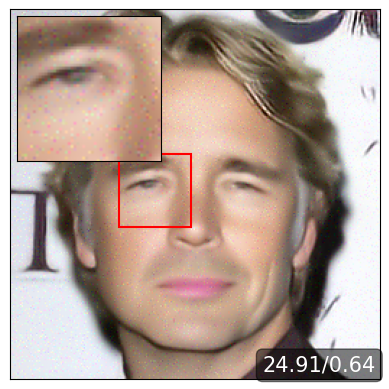}
	\caption{RED}
\end{subfigure}
\caption{Inpainting of an image with an inpainting ratio of 0.7 and noise level of 0. Visually, SOR produces more textures compared to RED which tends to average the image. However, it adds more artifacts.}\label{fig:inpainting_0.7_0}
\end{figure}

\subsection{Visual results with a greater noise level}
To further the comparisons, we tested PGD, SOR and RED on images with an increased noise level. These visual experiments show in particular how SOR is more stable compared to PGD and RED in more challenging contexts.
\subsubsection{Inpainting}
\begin{figure}
\centering
\begin{subfigure}{0.23\textwidth}
	\includegraphics[width=1\textwidth]{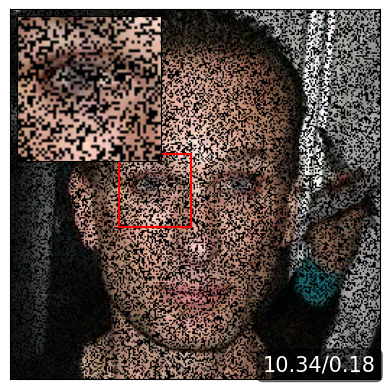}
	\caption{Observed}
\end{subfigure}
\hspace{0.02cm}
\begin{subfigure}{0.23\textwidth}
	\includegraphics[width=1\textwidth]{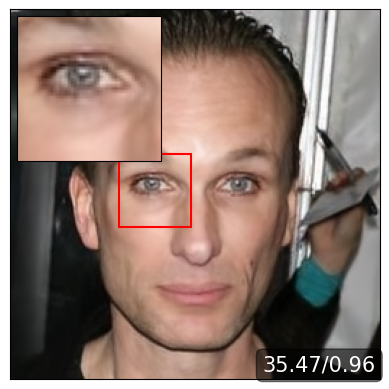}
	\caption{PGD}
\end{subfigure}
\hspace{0.02cm}
\begin{subfigure}{0.23\textwidth}
	\includegraphics[width=1\textwidth]{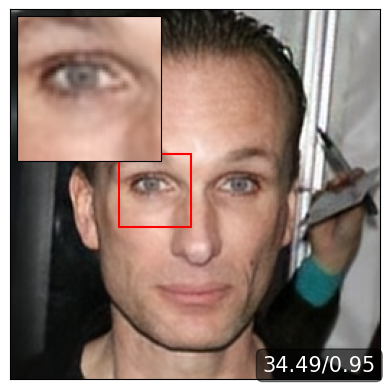}
	\caption{SOR}
\end{subfigure}
\hspace{0.02cm}
\begin{subfigure}{0.23\textwidth}
	\includegraphics[width=1\textwidth]{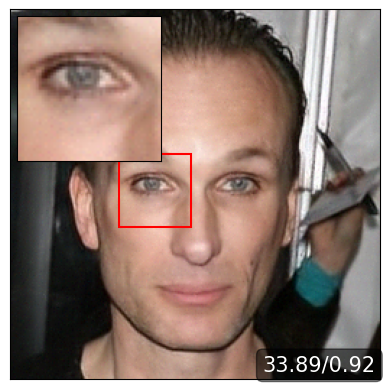}
	\caption{RED}
\end{subfigure}

\caption{Inpainting of an image with an inpainting ratio of 0.4 and noise level of 0.02}\label{fig:inpainting_0.4_0.02}
\end{figure}

\begin{figure}[!h]
\centering
\begin{subfigure}{0.23\textwidth}
	\includegraphics[width=1\textwidth]{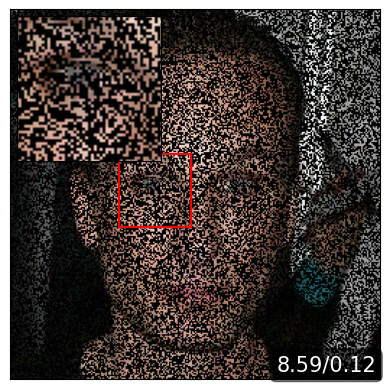}
	\caption{Observed}
\end{subfigure}
\hspace{0.02cm}
\begin{subfigure}{0.23\textwidth}
	\includegraphics[width=1\textwidth]{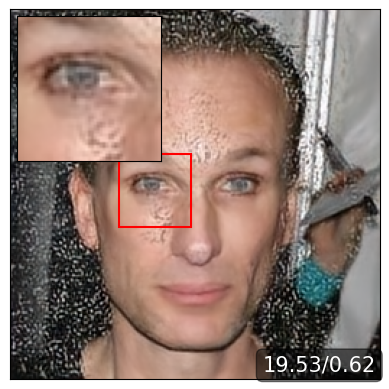}
	\caption{PGD}
\end{subfigure}
\hspace{0.02cm}
\begin{subfigure}{0.23\textwidth}
	\includegraphics[width=1\textwidth]{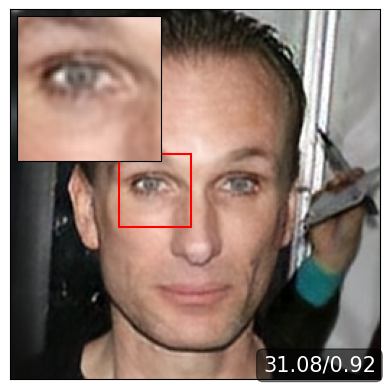}
	\caption{SOR}
\end{subfigure}
\hspace{0.02cm}
\begin{subfigure}{0.23\textwidth}
	\includegraphics[width=1\textwidth]{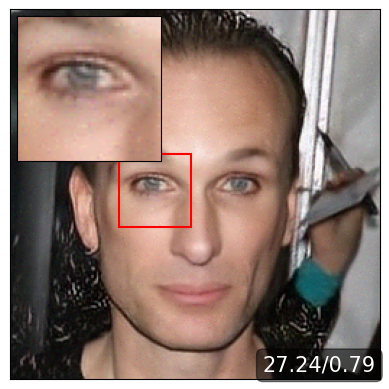}
	\caption{RED}
\end{subfigure}
\caption{Inpainting of an image with an inpainting ratio of 0.6 and noise level of 0.02}\label{fig:inpainting_0.6_0.02}
\end{figure}

\begin{figure}[!h]
\centering
\begin{subfigure}{0.23\textwidth}
	\includegraphics[width=1\textwidth]{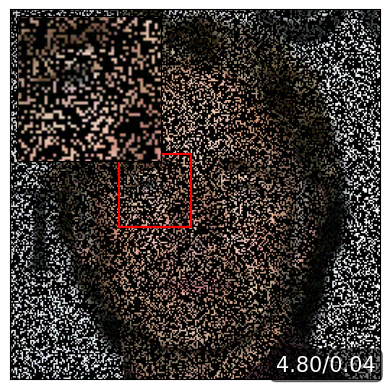}
	\caption{Observed}
\end{subfigure}
\hspace{0.02cm}
\begin{subfigure}{0.23\textwidth}
	\includegraphics[width=1\textwidth]{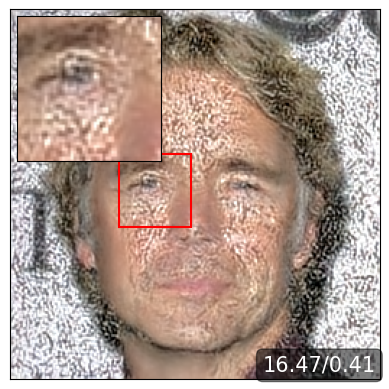}
	\caption{PGD}
\end{subfigure}
\hspace{0.02cm}
\begin{subfigure}{0.23\textwidth}
	\includegraphics[width=1\textwidth]{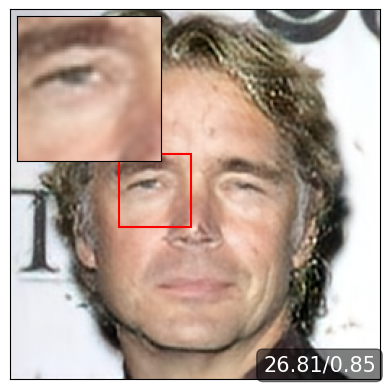}
	\caption{SOR}
\end{subfigure}
\hspace{0.02cm}
\begin{subfigure}{0.23\textwidth}
	\includegraphics[width=1\textwidth]{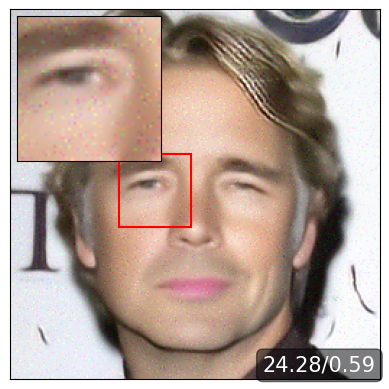}
	\caption{RED}
\end{subfigure}
\caption{Inpainting of an image with an inpainting ratio of 0.7 and noise level of 0.02}\label{fig:inpainting_0.7_0.02}
\end{figure}

\newpage 
\subsubsection{Super-resolution}
\begin{figure}[!h]
\centering
\begin{subfigure}{0.23\textwidth}
	\includegraphics[width=1\textwidth]{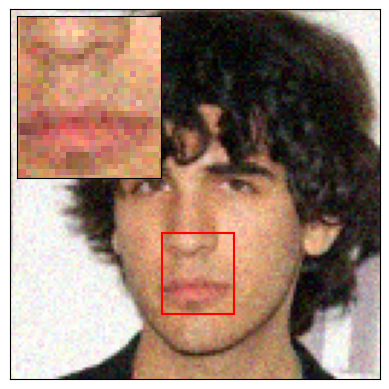}
	\caption{Observed}
\end{subfigure}
\hspace{0.02cm}
\begin{subfigure}{0.23\textwidth}
	\includegraphics[width=1\textwidth]{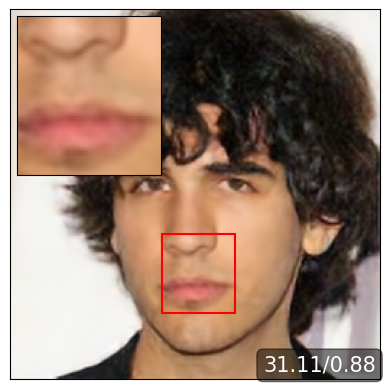}
	\caption{PGD}
\end{subfigure}
\hspace{0.02cm}
\begin{subfigure}{0.23\textwidth}
	\includegraphics[width=1\textwidth]{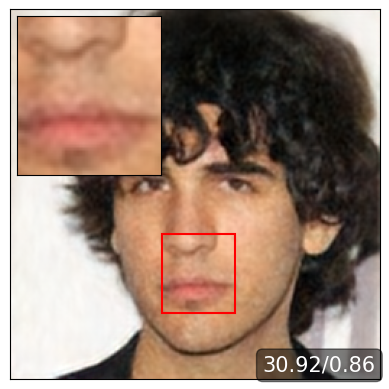}
	\caption{SOR}
\end{subfigure}
\hspace{0.02cm}
\begin{subfigure}{0.23\textwidth}
	\includegraphics[width=1\textwidth]{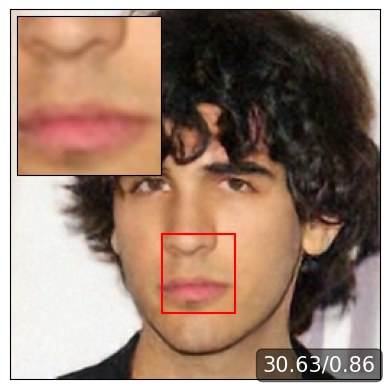}
	\caption{RED}
\end{subfigure}
\caption{Super-resolution of an image with a factor 2 and noise level of 0.05. All methods correctly recover the original image, with PGD being slightly the best.}
\end{figure}

\begin{figure}[!h]
\centering
\begin{subfigure}{0.23\textwidth}
	\includegraphics[width=1\textwidth]{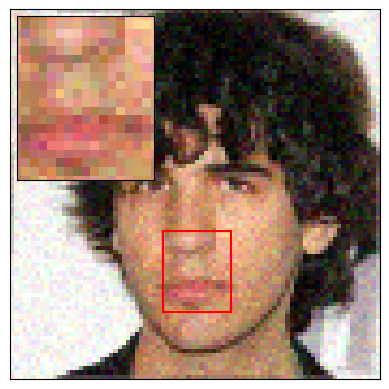}
	\caption{Observed}
\end{subfigure}
\hspace{0.02cm}
\begin{subfigure}{0.23\textwidth}
	\includegraphics[width=1\textwidth]{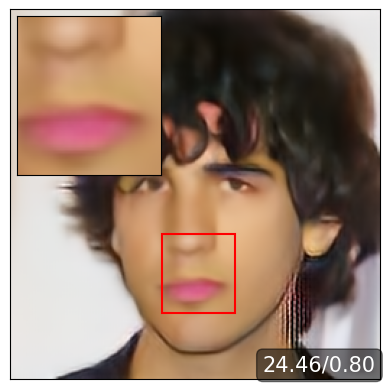}
	\caption{PGD}
\end{subfigure}
\hspace{0.02cm}
\begin{subfigure}{0.23\textwidth}
	\includegraphics[width=1\textwidth]{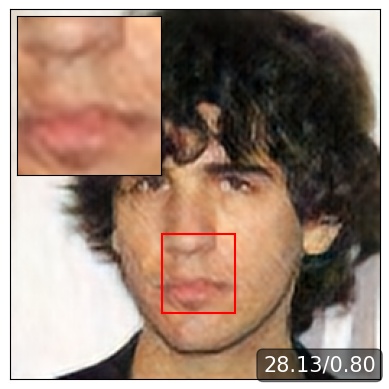}
	\caption{SOR}
\end{subfigure}
\hspace{0.02cm}
\begin{subfigure}{0.23\textwidth}
	\includegraphics[width=1\textwidth]{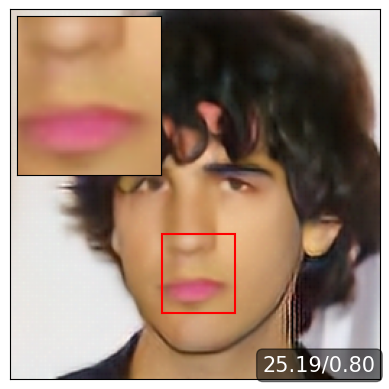}
	\caption{RED}
\end{subfigure}
\caption{Super-resolution of an image with a factor 3 and noise level of 0.05. All methods correctly recover the original image, with PGD being slightly the best.}
\end{figure}

\newpage 
\subsubsection{Deblurring}

~~
\begin{figure}[!h]
\centering
\hspace{0.02cm}
\begin{subfigure}{0.23\textwidth}
	\includegraphics[width=1\textwidth]{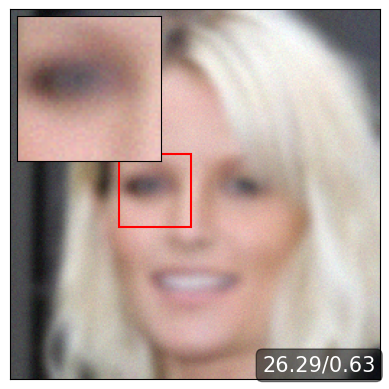}
	\caption{Observed}
\end{subfigure}
\hspace{0.02cm}
\begin{subfigure}{0.23\textwidth}
	\includegraphics[width=1\textwidth]{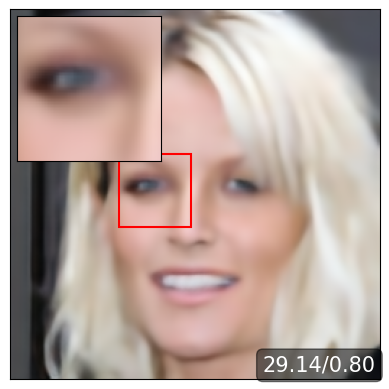}
	\caption{PGD}
\end{subfigure}	
\hspace{0.02cm}
\begin{subfigure}{0.23\textwidth}
	\includegraphics[width=1\textwidth]{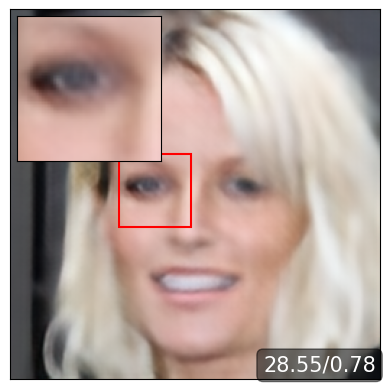}
	\caption{SOR}
\end{subfigure}
\hspace{0.02cm}
\begin{subfigure}{0.23\textwidth}
	\includegraphics[width=1\textwidth]{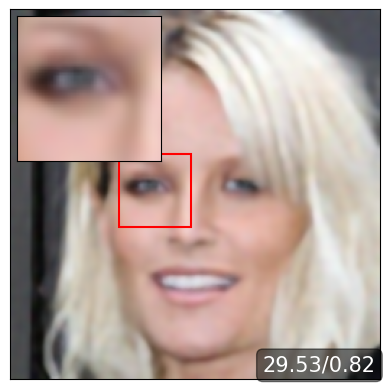}
	\caption{RED}
\end{subfigure}	
\caption{Deblurring of an image of a 5$\times$5 Gaussian kernel and noise level of 0.02.}
\end{figure}

\end{document}